\documentclass[twocolumn,twoside]{IEEEtran}

\IEEEoverridecommandlockouts                              
                                                        
\usepackage{amsmath,amsfonts}
\usepackage[ruled, vlined]{algorithm2e}
\usepackage{algpseudocode}
\usepackage{amsmath} 
\usepackage{array}
\usepackage{subfigure}
\usepackage{textcomp}
\usepackage{stfloats}
\usepackage{url}
\usepackage{verbatim}
\usepackage{graphicx}
\usepackage{cite}

\usepackage{amsthm,amssymb}
\usepackage{color}
\usepackage{pifont}
\usepackage{hyperref}
\usepackage{footnote}

\usepackage{enumitem}
\usepackage{diagbox}
\usepackage{booktabs}
\usepackage{multirow}
\usepackage{pifont}
\usepackage{tablefootnote}
\usepackage{ragged2e}
\hypersetup{
colorlinks=true,
linkcolor=red,
citecolor=black,
}
\usepackage{flushend}

\newcommand{\Bc}{\mathcal{B}}
\newcommand{\Cc}{\mathcal{C}}

\newcommand{\Ec}{\mathcal{E}}

\newcommand{\Gc}{\mathcal{G}}

\newcommand{\Mc}{\mathcal{M}}
\newcommand{\Nc}{\mathcal{N}}

\newcommand{\Qc}{\mathcal{Q}}

\newcommand{\Sc}{\mathcal{S}}

\newcommand{\Uc}{\mathcal{U}}

\newcommand{\Xc}{\mathcal{X}}

\newcommand\dbf{\mathbf{d}}

\newcommand\mbf{\mathbf{m}}

\newcommand\pbf{\mathbf{p}}
\newcommand\qbf{\mathbf{q}}

\newcommand\vbf{\mathbf{v}}

\newcommand\xbf{\mathbf{x}}

\newcommand\Abf{\mathbf{A}}

\newcommand\Dbf{\mathbf{D}}

\newcommand\Mbf{\mathbf{M}}

\newcommand\Rbf{\mathbf{R}}

\newcommand\Xbf{\mathbf{X}}

\newcommand{\real}{\ensuremath{\mathbb{R}}}

\newcommand\phib{\boldsymbol{\phi}}

\newcommand{\ones}{\mathbf{1}}

\newcommand{\proj}{\text{Proj}}

\newtheorem{theorem}{Theorem}[section]
\newtheorem{proposition}[theorem]{Proposition}

\theoremstyle{definition}
\newtheorem{remark}[theorem]{Remark}

\newtheorem{problem}[theorem]{Problem}

\newcommand{\oprocendsymbol}{\hbox{$\bullet$}}
\newcommand{\oprocend}{\relax\ifmmode\else\unskip\hfill\fi\oprocendsymbol}

\newcommand{\longthmtitle}[1]{\mbox{}{\textit{(#1):}}}

\IEEEaftertitletext{\vspace{-1\baselineskip}} 

\makeatletter
\def\footnoterule{\kern-3\p@
  \hrule \@width 2in \kern 2.6\p@} 
\makeatother

\title{Stability Constrained Voltage Control in Distribution Grids with Arbitrary Communication Infrastructure}

\author{Zhenyi Yuan,~\IEEEmembership{Graduate Student Member, IEEE}, Jie Feng,~\IEEEmembership{Graduate Student Member, IEEE}, \\
Yuanyuan Shi,~\IEEEmembership{Member, IEEE}, and Jorge Cort\'es,~\IEEEmembership{Fellow, IEEE}
\thanks{This work was supported by NSF Award ECCS-1947050 and ECCS-2200692. The work of J. Feng was supported by the UC-National Laboratory In Residence Graduate Fellowship L24GF7923. }
\thanks{Z. Yuan and J. Cort\'es are with the Department of Mechanical and Aerospace Engineering, University of California, San Diego, La Jolla, CA 92093, USA,
{\tt\small \{z7yuan,cortes\}@ucsd.edu}. J. Feng and Y. Shi are with the Department of Electrical and Computer Engineering, University of California, San Diego, La Jolla, CA 92093, USA,
{\tt\small \{jif005,yyshi\}@ucsd.edu}.}%
}

\begin{document}

\maketitle


\begin{abstract}
We consider the problem of designing learning-based reactive power controllers that perform voltage regulation in distribution grids while ensuring closed-loop system stability. In contrast to existing methods, where the provably stable controllers are restricted to be decentralized, we propose a unified design framework that enables the controllers to take advantage of an \emph{arbitrary} communication infrastructure on top of the physical power network. This allows the controllers to incorporate information beyond their local bus, covering existing methods as a special case and leading to less conservative constraints on the controller design. We then provide a design procedure to construct
input convex neural network (ICNN) based controllers that satisfy the identified stability constraints by design under arbitrary communication scenarios, and train these controllers using supervised learning. Simulation results on the the University of California, San Diego (UCSD) microgrid testbed illustrate the effectiveness of the framework and highlight the role of communication in improving control performance. 
\end{abstract}


\section{Introduction}

The increasing penetration of distributed energy resources (DERs) poses great challenges to system operations and stability of distribution grids (DGs). 
For instance, the intermittence of renewable energy sources can lead to rapid and unpredictable fluctuations in the load and generation profiles, causing large voltage variations. This has motivated the study of Volt/Var control strategies which aim to regulate voltages within acceptable preassigned limits by commanding DERs' reactive power injections~\cite{PS-RH-VJN-VV-AMA-AKS:23}. However, without appropriate design, the Volt/Var controllers may render the closed-loop system unstable, resulting in voltage oscillations. In this paper, we propose a unified framework to design Volt/Var control strategies that are amenable to DGs with various communication topologies and rigorously guarantee the stability of the closed-loop system.

\subsubsection*{Literature Review}
Initial attempts to ensure closed-loop stability with Volt/Var controllers start from decentralized designs~\cite{NL-GQ-MD:14, HZ-HJL:15, GC-RC:17, XZ-MF-ZL-LC-SHL:21}, and is also recently revisited by~\cite{SG-AMS-SC-VK:23,IM-SG-SC-VK:24,AC-EP-MG-ED:25}, where the monotonicity of the decentralized controllers is identified as the key property to ensure the asymptotic stability. However, the controller form in these works is constrained to be piece-wise linear, as advocated by the IEEE 1547 Standard~\cite{IEEE1547}. Under these designs, a certain level of steady-state optimality can be achieved only for specific cost function forms.
Moreover, the piece-wise linear form of the controllers further prohibits the improvement of performance, and does not take advantage of the fact that inverter-based DERs can implement almost arbitrary control laws thanks to the flexibility in power electronic interfaces. 

Motivated by the above factors, recent works~\cite{YS-GQ-SL-AA-AW:22, WC-JL-BZ:22, JF-YS-GQ-SHL-AA-AW:24, JF-WC-JC-YS:23-csl} generalize the form of decentralized Volt/Var controller from piece-wise linear to nonlinear functions parameterized by neural networks, which is more flexible, and still ensures asymptotic stability as long as it satisfies the same monotonicity property. However, these works ask the controllers to be continuously differentiable and assume there are no reactive power regulation constraints. The works~\cite{GC-ZY-MKS-JC:22-cdc, ZY-GC-MKS-JC:24-tps} improve them by relaxing the differentiability requirements and taking reactive power regulation constraints into account, followed by~\cite{ZY-GC-JC:23-csl}, which explores other types of constraints on controllers to ensure closed-loop system stability.
Based on these results, machine learning techniques, e.g., reinforcement learning~\cite{YS-GQ-SL-AA-AW:22, WC-JL-BZ:22, JF-YS-GQ-SHL-AA-AW:24, JF-WC-JC-YS:23-csl}, supervised learning~\cite{GC-ZY-MKS-JC:22-cdc, ZY-GC-MKS-JC:24-tps, ZY-GC-JC:23-csl}, and unsupervised learning~\cite{ZY-GC-AZ-JC:24}, paired with neural network parameterized controllers that satisfy the stability constraints by design, are employed to enhance the optimality of controllers. We remark that~\cite{YS-GQ-SL-AA-AW:22, WC-JL-BZ:22, JF-YS-GQ-SHL-AA-AW:24, JF-WC-JC-YS:23-csl} focus on transient performance optimization, while~\cite{GC-ZY-MKS-JC:22-cdc, ZY-GC-MKS-JC:24-tps, ZY-GC-JC:23-csl, ZY-GC-AZ-JC:24} focus on steady-state performance.

While tremendous progress has been made, the above works all employ \emph{decentralized} designs. DGs are often equipped with a communication infrastructure on top of the physical network, which allows some buses to exchange information with their neighbors. Motivated by the fact that the optimal power flow (OPF) problem can be solved in a \emph{distributed} fashion under appropriate connectivity conditions on the communication infrastructure, e.g.,~\cite{QP-SHL:16}, here we explore distributed designs to enhance the performance of decentralized voltage controllers. 
Relevant results include~\cite{SB-SZ:13, BZ-AYSL-ADDG-DT:14, SB-RC-GC-SZ:15, GC-RC:17, ED-AS:18, GQ-NL:20}, where distributed optimization-based feedback controllers are derived to iteratively steer the network toward OPF solutions. Theoretical analysis characterizes the convergence properties of these distributed algorithms, establishing the asymptotic stability of the closed-loop system. However, these algorithms  require that the DGs are endowed with a reliable real-time communication network that is connected and consistent with the physical network, which is rarely satisfied in practice for DGs.
While decentralized designs may incorporate optimality considerations, there is significant potential for performance improvement by leveraging available communication infrastructure that may not be consistent with the physical network (or even connected). Hence, instead of asking the controllers to recursively communicate over a restrictive communication network to converge to OPF solutions, here we study the problem of how to take advantage of an arbitrarily existing communication infrastructure to improve performance of the decentralized Volt/Var controllers in DGs by having access to more measurements, so that they behave closer to OPF solutions, and yet ensure closed-loop stability.

\begin{figure*}
    \centering    \includegraphics[width=0.98\linewidth]{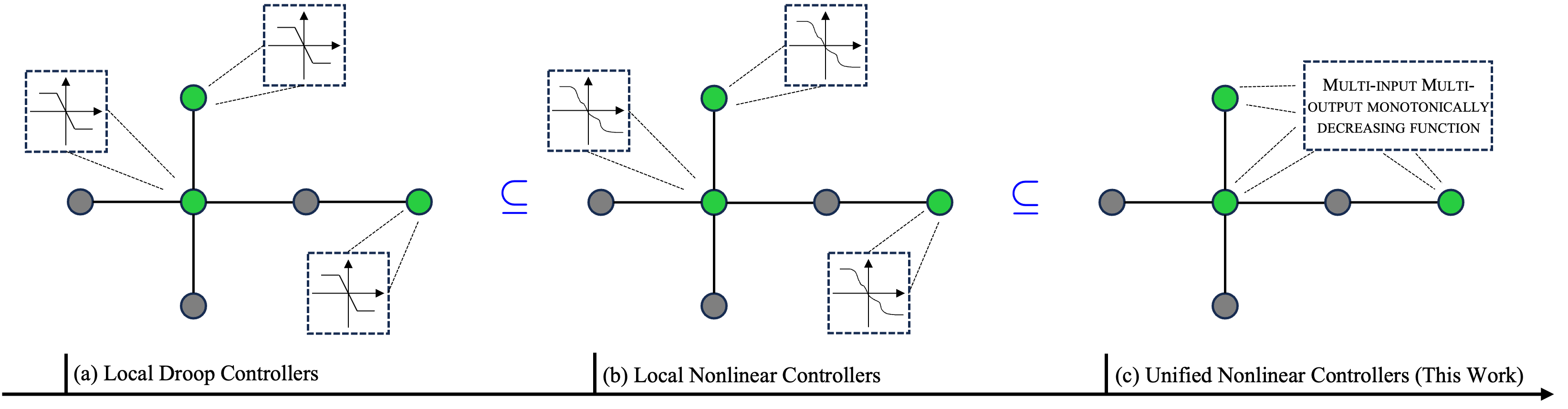}     
    \caption{Road map of the development of Volt/Var control methods with stability guarantees. The green nodes represent controllable DERs in the DG.
    }
    \label{fig:development}
\end{figure*}

\subsubsection*{Statement of Contributions}
In this paper, we propose 
a unified framework to design provably stable Volt/Var controllers for DGs with arbitrary communication infrastructures. 
The novelty of this paper is summarized as 
\begin{itemize}
    \item Compared to existing results on Volt/Var controller design to ensure closed-loop system stability, e.g.~\cite{JF-YS-GQ-SHL-AA-AW:24,ZY-GC-MKS-JC:24-tps}, where the \emph{individual} controller at each bus is required to be monotone, we identify a less-conservative stability condition by studying a new contraction metric that only requires all controllers in the network to \emph{collectively} satisfy the monotonicity constraint. Thereby, existing results can be viewed as special cases of the method proposed in this paper, see Fig.~\ref{fig:development} for illustration;
    \item We propose a systematic approach to parameterize the local-acting Volt/Var controllers, such that they can collectively satisfy the identified network-wide monotonicity constraint,
    using input convex neural networks (ICNNs). This framework is novel and unified, in the sense that it applies to arbitrary communication networks the DG is endowed with, and guarantees stability by design without requiring case-by-case adaptation to the communication infrastructure nor additional post-processing;
    \item We validate the proposed framework on a real-world testbed - the University of California, San Diego (UCSD) microgrid - using operational data collected from the campus. This evaluation on a realistic distribution system demonstrates the practical effectiveness of the proposed approach and highlights the critical role of communication in enhancing the performance of learned Volt/Var controllers for voltage control.
\end{itemize}

\subsubsection*{Notation}
Throughout the paper, $\real$ denotes the set of real numbers. 
Upper and lower case boldface letters denote matrices and column vectors, respectively. 
Given a matrix $\Abf$, $\Abf \succ~0$ denotes that matrix $\Abf$ is positive definite. 
We use $\|\cdot\|$ to represent the Euclidean norm, and $|\cdot|$ the cardinality when the argument is a set. We let $\|\xbf\|_{\Abf} = \xbf^\top \Abf \xbf$ if $\Abf$ is square.
The symbol $(\cdot)^\top$ stands for transposition, 
and $\ones$ denotes vector of all ones with appropriate dimensions.
Operator $\proj_{\Xc}(\cdot)$ represents the projection of the argument into the set $\Xc$.


\section{Preliminaries and Problem Formulation}\label{sec:preliminaries}
In this section, we introduce the DG model and the communication network, and formulate the problem of interest.

\subsection{Distribution Grid Modeling}\label{subsec:grid_modeling}

A radial single-phase (or a balanced three-phase) DG having $N+1$ buses can be modeled by a tree graph $\Gc=(\Nc,\Ec)$ rooted at the substation. The nodes in $\Nc_0:=\{0,\ldots,N\}$ are associated with grid buses, and the edges in $\Ec$ with lines. The substation node, labeled as 0, behaves as an ideal voltage source imposing the nominal voltage of 1 p.u. and we let $\Nc = \Nc_0 \setminus \{0\}$. 
We use $\Mbf^0 = [\mbf_0^\top ; \Mbf] \in \real^{(N+1) \times N}$ to denote the incidence matrix of the graph $\Gc$, where $\mbf_0^\top$ is the row corresponds to the substation node.
The voltage magnitude at bus $n\in \Nc$ is denoted as $v_n\in \real$, and the active and reactive power injections at bus $n\in \Nc$ are $p_n,q_n\in \real$, respectively. Powers take positive (negative) values, i.e., $p_n, q_n \geq 0$ ($p_n, q_n \leq 0$), when they are \emph{injected into} (\emph{absorbed from}) the grid.
The vectors $ \vbf,\pbf, \qbf \in \real^N$ collect the voltage magnitudes, active and reactive power injections for buses $1,2,\dots N$.
Let $(m,n)$ be an edge in $\Ec$, let $r_{mn}$ and $x_{mn}$ denote its resistance and reactance, and $P_{mn}$ and $Q_{mn}$ the real and reactive power from bus $m$ to $n$, respectively.
For every line $(m,n) \in \Ec$, according to the DistFlow equations~\cite{MEB-FFW:89}, the power flow model is 
\begin{subequations}\label{eq:nonlinear_pf}
\begin{align}
P_{mn}-\sum_{(n,k) \in \Ec} P_{n k}= & -p_n+r_{mn} \frac{P_{mn}^2+Q_{mn}^2}{v_m^2},
\\
Q_{mn}-\sum_{(n,k) \in \Ec} Q_{n k}= & -q_n+x_{mn} \frac{P_{mn}^2+Q_{mn}^2}{v_m^2},
\\
v_m^2-v_n^2= & 2\left(r_{mn} P_{mn}+x_{mn} Q_{mn}\right) \notag \\
& -\left(r_{mn}^2+x_{mn}^2\right) \frac{P_{mn}^2+Q_{mn}^2}{v_m^2}.
\end{align}
We note that if the DG is connected, and there is no zero-impedance branch, then the DistFlow model is equivalent to the standard bus injection model~\cite{BS-SHL-KMC:12,TD-RL-YY-FB:19}.
\end{subequations}

\subsection{Communication Network}\label{subsec:comm_graph}
We consider the case where an undirected communication network exists in the DG. 
We say bus $n$ is a neighbor of bus $m$ if bus $m$ can receive the information broadcasted from bus $n$.
We use $\Nc_i$ to represent the set of neighbors of bus $i \in \Nc$. We note that the communication network
can be different from the physical network (e.g., buses $m$ and $n$ might be physically connected, but not neighbors of each other in the communication network). In fact, this is a common occurrence due to the fact that DGs are frequently reconfigured~\cite{DD-VK-GC:24}.
Moreover, we allow the communication network to be arbitrary and not necessarily connected, i.e., we allow the case that only partial communication is available or even the case that there is no communication at all. Hence, our setting covers all possible communication scenarios in DGs, including, e.g., full communication, distributed communication, and no communication cases, see Fig.~\ref{fig:communication_setups}.
\begin{figure}[tb]
    \centering
    \subfigure[Full communication]{
        \includegraphics[scale=0.13]{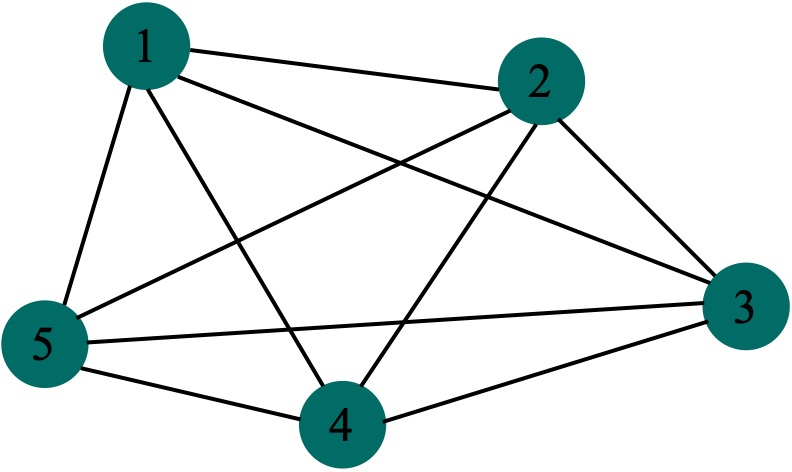}
        \label{fig:communication_setups_centralized}
    }\hspace{2ex}
    \subfigure[Distributed communication-I]{
	\includegraphics[scale=0.13]{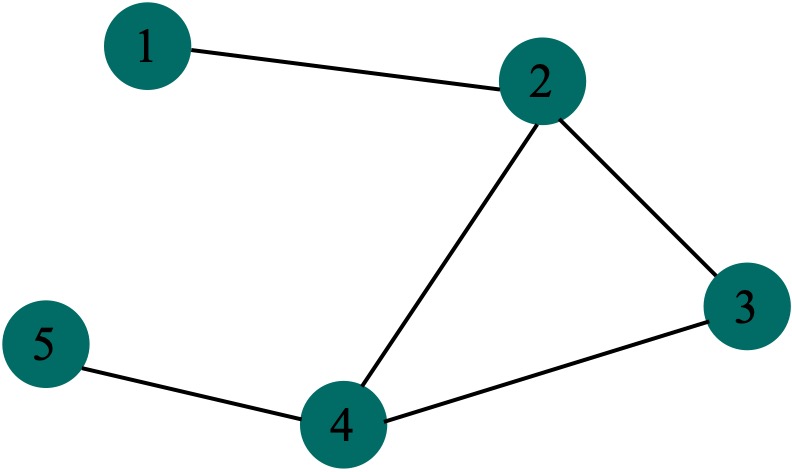}
 \label{fig:communication_setups_distributed}
    }
    \subfigure[Distributed communication-II]{
	\includegraphics[scale=0.13]{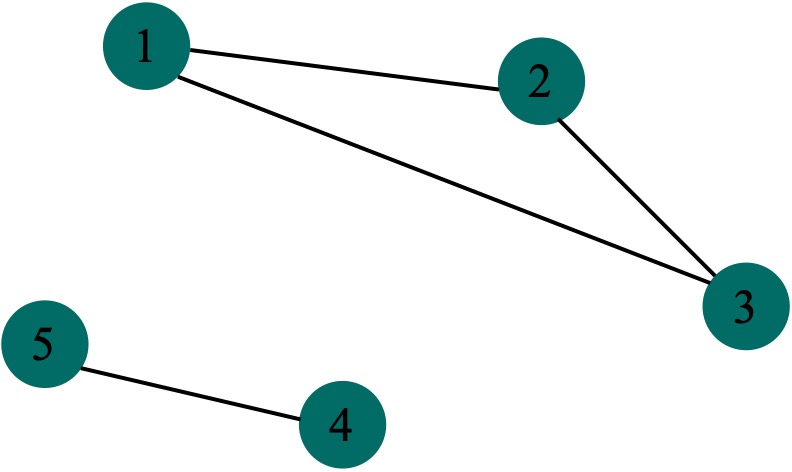}
 \label{fig:communication_setups_hybrid}
    }\hspace{2ex}
    \subfigure[No communication]{
	\includegraphics[scale=0.13]{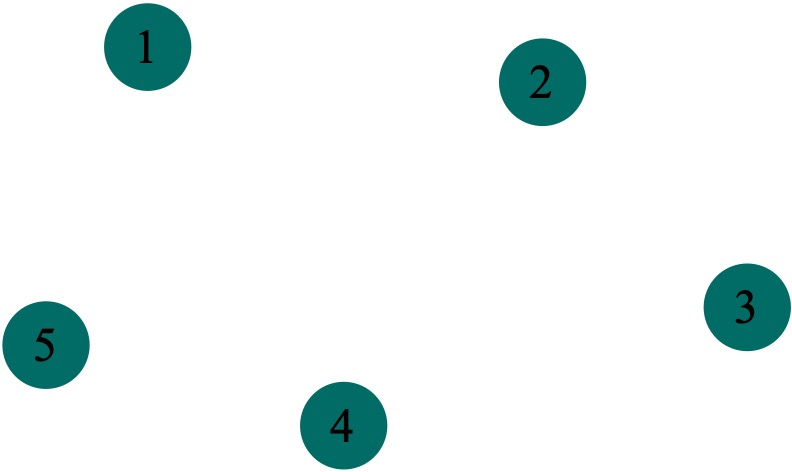}
 \label{fig:communication_setups_decentralized}
    }
  \caption{Examples of different communication infrastructure in DGs. Note that the difference between (b) and (c) is that the former is connected and the latter is not.}
  \label{fig:communication_setups}
\end{figure}

\subsection{Problem Formulation}\label{subsec:problem}

We are interested in dynamically adjusting the reactive power injections $\{q_i\}_{i \in \Nc}$ in~\eqref{eq:nonlinear_pf} for the inverter-interfaced DERs to perform voltage regulation in DGs. Specifically, we aim to design, for each $i \in \Nc$, a function
\begin{align}
    \phi_i: \real^{|\Nc_i| + 1} \rightarrow \Qc_i, ~ (v_i, \{v_j\}_{j \in \Nc_i}) \mapsto  \phi_i(v_i, \{v_j\}_{j \in \Nc_i})
\end{align}
that takes accessible voltage magnitude measurements as input and determines its reactive power injection, where $\Qc_i := \{q_i ~ | ~ q_{i,\min} \leq q_i \leq q_{i,\max}\}$ defines the reactive power control capability of bus~$i$.

A first consideration for this type of design is how to ensure the stability of the closed-loop dynamics~\eqref{eq:nonlinear_pf} resulting from dynamically adjusting the reactive power injections $\{q_i\}_{i \in \Nc}$ via the feedback mechanisms $\{\phi_i\}_{i \in \Nc}$.
Related works in the literature, e.g.~\cite{YS-GQ-SL-AA-AW:22,WC-JL-BZ:22,JF-YS-GQ-SHL-AA-AW:24,JF-WC-JC-YS:23-csl,GC-ZY-MKS-JC:22-cdc,ZY-GC-MKS-JC:24-tps,ZY-GC-JC:23-csl,ZY-GC-AZ-JC:24}, consider the case where $\{\phi_i\}_{i \in \Nc}$ are \emph{decentralized}, i.e., each $\phi_i$ is a function solely depending on $v_i$. In general, little is known about how to ensure closed-loop stability for the case when the functions $\{\phi_i\}_{i \in \Nc}$ are not decentralized and employ other communication infrastructure, cf. Fig.~\ref{fig:communication_setups}. A second consideration is to have the reactive power setpoints determined by $\{\phi_i\}_{i \in \Nc}$ approximate the OPF solutions for optimality considerations. While solving the OPF problem requires restrictive communication and massive computational resources, it is unrealistic to assign the outputs $\{\phi_i\}_{i \in \Nc}$ by solving the OPF instances online. Machine learning techniques are often used to find optimal $\{\phi_i\}_{i \in \Nc}$ offline~\cite{YS-GQ-SL-AA-AW:22,WC-JL-BZ:22,JF-YS-GQ-SHL-AA-AW:24,JF-WC-JC-YS:23-csl,GC-ZY-MKS-JC:22-cdc,ZY-GC-MKS-JC:24-tps,ZY-GC-JC:23-csl,ZY-GC-AZ-JC:24}. 
Therefore, it is of interest to figure out how to formulate the learning problem to synthesize desired Volt/Var controllers mimicking the OPF solutions while guaranteeing closed-loop stability.
Below we formalize these problems.

\begin{problem}\longthmtitle{Stability constrained voltage control}
    Consider the DistFlow model~\eqref{eq:nonlinear_pf} together with $\{\phi_i\}_{i \in \Nc}$ to determine the reactive power setpoints, our goal is to answer the following two questions:
    \begin{itemize}
        \item[(Q1)] \textit{What conditions on $\{\phi_i\}_{i \in \Nc}$ are sufficient to guarantee the closed-loop system to be asymptotically stable under arbitrary communication infrastructures?}
        \item[(Q2)] \textit{How to formulate the learning task to find optimal $\{\phi_i\}_{i \in \Nc}$ that mimic the OPF solutions while 
        satisfying the stability conditions?}
    \end{itemize}
\end{problem}
In the following, we answer (Q1) and (Q2) in Section~\ref{sec:stability} and Section~\ref{sec:learning}, respectively.

\section{Unified Stability Conditions on Volt/Var Controllers}\label{sec:stability}
We consider the following incremental control law for steering the reactive power setpoint to the output of $\phi_i$ for all $i \in \Nc$:
\begin{align}\label{eq:bus_react_upd} 
    q_i(t+1) = q_i(t) + \epsilon \Big(\phi_i(v_i(t), \{v_j(t)\}_{j \in \Nc_i}) - q_i(t) \Big).
\end{align}
The parameter $\epsilon \in [0,1]$ is the stepsize to steer $q_i$ to the output of $\phi_i$. Note that this rule is implementable in a distributed fashion over the available communication network. We also remark that given that the output of $\phi_i$ is constrained to $\Qc_i$ for all $i \in \Nc$, if $q_i(0) \in \Qc_i$, then we have $q_i(t) \in \Qc_i$ for all $t \geq 0$ as the new reactive power setpoint would be the convex combination of two points belong to $\Qc_i$. Throughout the paper, we assume that $q_i(0) \in \Qc_i$ for all $i \in \Nc$.

Similar to existing works~\cite{ED-AS:18,GQ-NL:20,HZ-HJL:15}, we consider the linearized version (LinDistFlow)~\cite{HZ-HJL:15} of~\eqref{eq:nonlinear_pf} for ease of theoretical analysis of system stability:
\begin{subequations}\label{eq:linear_pf}
    \begin{align}
P_{mn}-\sum_{(n,k) \in \Ec} P_{n k} & =-p_n, \\
Q_{mn}-\sum_{(n,k) \in \Ec} Q_{n k} & =-q_n, \\
v_m-v_n & =r_{mn} P_{mn}+x_{mn} Q_{mn},
\end{align}
\end{subequations}
which can be written in the following compact form
 \begin{align}\label{eq:v=Rp+Xq}
 \vbf = \Rbf \pbf + \Xbf \qbf + \ones.
 \end{align}
Here, $\Rbf = \Mbf^{-\top} \Dbf_r \Mbf^{-1}$ and $\Xbf = \Mbf^{-\top} \Dbf_x \Mbf^{-1}$, where $\Dbf_r, \Dbf_x \in \real^{N \times N}$ are diagonal matrices which respectively collect $r_{mn}$ and $x_{mn}$ for all $(m,n) \in \Ec$. 

Collecting $\{\phi_i\}_{i \in \Nc}$ into $\phib$, and adopting the linearized power flow equation~\eqref{eq:v=Rp+Xq}, we obtain the closed-loop system dynamics in the following compact form:
\begin{subequations}\label{eq:sys_dyn}
  \begin{align}
    \qbf(t+1) &= \qbf(t) + \epsilon \left(\phib(\vbf(t)) - \qbf(t) \right), \label{eq:sys_dyn_q}  \\
    \vbf(t+1) &= \Xbf \qbf(t+1) + \hat \vbf, \label{eq:sys_dyn_v} 
\end{align}  
\end{subequations}
where it is tacitly assumed that the variability of $\pbf$ is at a much slower timescale than the above iterates, thus the $\pbf$ remains a constant term during the convergence of the iterations of~\eqref{eq:sys_dyn}.
leading to a fixed $\hat \vbf = \Rbf \pbf + \ones$ in convergence analysis of~\eqref{eq:sys_dyn}.

From~\eqref{eq:sys_dyn}, it is easy to see that an equilibrium point $(\vbf^\star,\qbf^\star)$ should satisfy 
    \begin{align}\label{eq:equilibria}
        \phib(\vbf^\star) &= \qbf^\star,~\vbf^\star = \Xbf \qbf^\star + \hat \vbf,
    \end{align}
and thus, consistently with our prior work~\cite{GC-ZY-MKS-JC:22-cdc,ZY-GC-MKS-JC:24-tps}, we refer $\{\phi_i\}_{i \in \Nc}$ to as \emph{equilibrium functions}.

\begin{remark}\longthmtitle{Local asymptotic stability of the original nonlinear dynamics}
   In the closed-loop stability analysis here, we rely on the linearized power flow model~\eqref{eq:sys_dyn_v}, which is common in the existing works on reactive power control, see e.g.~\cite{NL-GQ-MD:14,HZ-HJL:15,GC-RC:17,XZ-MF-ZL-LC-SHL:21}. We emphasize that, although the linearization of the power flow model adds limitations to the theoretical results itself, the global asymptotic stability of the closed-loop system with linearized power flow model (as we will show later) still implies local asymptotic stability of the closed-loop system with original nonlinear power flow model in~\eqref{eq:nonlinear_pf}, cf.~\cite[Theorem 3.3]{NB-RC-LS:18}. We leave the characterization of the region of attraction for future work.
   \oprocend
\end{remark}

\subsection{Existence and Uniqueness of the Equilibrium}
Here, we study the properties of equilibrium points of the closed-loop system~\eqref{eq:sys_dyn}.
The following result identifies a sufficient condition on these equilibrium functions so that there exists a unique equilibrium point.
\begin{proposition}\longthmtitle{Existence and uniqueness of the equilibrium}\label{prop:unique_equi}
    Consider the closed-loop system~\eqref{eq:sys_dyn}. If $\phib$ is a monotonically decreasing function w.r.t. $\vbf$, i.e.,
    \begin{align}\label{eq:monotonicity}
        \left( \phib(\vbf) - \phib(\vbf^\prime)  \right)^\top \left( \vbf - \vbf^\prime \right) \leq 0,
    \end{align}
    or equivalently,
    \begin{align}\label{eq:monotonicity-separate}
        \sum_{i \in \Nc} \left( \phi_i(v_i, \{v_j\}_{j \in \Nc_i}) - \phi_i(v_i^\prime, \{v_j^\prime\}_{j \in \Nc_i}) \right) (v_i - v_i^\prime) \leq 0
    \end{align}
    holds $\forall~\vbf,\vbf^\prime \in \real^{N}$, then~\eqref{eq:sys_dyn} admits a unique equilibrium.
\end{proposition}
\begin{proof}
    We first show the existence of the equilibrium. From~\eqref{eq:equilibria}, the closed-loop system~\eqref{eq:sys_dyn} admits an equilibrium if there exists a fixed point to
    \begin{align}
        \qbf = \phib(\Xbf \qbf + \hat \vbf) = \tilde \phib (\qbf).
    \end{align}
    where $\tilde \phib: \Qc \mapsto \Qc$. According to Brouwer's Fixed Point Theorem~\cite[Corollary 6.6]{KCB:85}, since $\Qc$ is convex and compact, such fixed point must exist.
    
    Now we show the uniqueness of the equilibrium, and we reason by contradiction. Suppose there exist two equilibrium points for~\eqref{eq:sys_dyn}, namely $(\qbf^*,\vbf^*)$ and $(\qbf^\sharp,\vbf^\sharp)$, with $\vbf^* \neq \vbf^\sharp$.
    From~\eqref{eq:equilibria}, $\phib(\vbf^*) = \qbf^*$ and $\phib(\vbf^\sharp) = \qbf^\sharp$ hold. Therefore,
    \begin{align}\label{eq:contradiction_1}
         \qbf^* - \qbf^\sharp = \phib(\vbf^*) - \phib(\vbf^\sharp).
    \end{align}
    From this, we deduce that
    \begin{align}
        (\vbf^* \!-\! \vbf^\sharp)^\top (\qbf^* \!-\! \qbf^\sharp) \!=\! (\vbf^* \!-\! \vbf^\sharp)^\top \left( \phib(\vbf^*) \!-\! \phib(\vbf^\sharp) \right) \! \leq \! 0,
    \end{align}
    where we have leveraged the monotonicity property~\eqref{eq:monotonicity}.    
    On the other hand, it follows from~\eqref{eq:equilibria} that $\vbf^* - \vbf^\sharp = \Xbf(\qbf^* - \qbf^\sharp)$, and thus
    \begin{align}\label{eq:contradiction_2}
        \qbf^* - \qbf^\sharp = \Xbf^{-1} (\vbf^* - \vbf^\sharp),
    \end{align}
    which implies that  
            \begin{align}
            (\vbf^* \!-\! \vbf^\sharp)^\top (\qbf^* \!-\! \qbf^\sharp) \!=\! (\vbf^* \!-\! \vbf^\sharp)^\top \Xbf^{-1} (\vbf^* -    \vbf^\sharp) > 0, 
    \end{align}
    leading to a contradiction.
\end{proof}

\subsection{Global Asymptotic Stability of the Equilibrium}

Given the fact that the closed-loop system~\eqref{eq:sys_dyn} admits a unique equilibrium under Proposition~\ref{prop:unique_equi}, we provide in the next result a sufficient condition on $\epsilon$ such that this unique equilibrium point is globally asymptotically stable.

\begin{theorem}\longthmtitle{Global asymptotic stability of the equilibrium}\label{thm:aymp_stability}
Consider the closed-loop system~\eqref{eq:sys_dyn} and let $\phib$ be monotonically decreasing function w.r.t. $\vbf$, cf.~\eqref{eq:monotonicity}. If $\epsilon$ satisfies
\begin{align}\label{eq:cond_epsilon}
    \epsilon < \min\left\{1, \frac{2}{1+L^2\|\Xbf\|^2}\right\},
\end{align}
where $L = \max \limits_{x,y \in \Bc} \frac{\|\phib(x) - \phib(y)\|}{\|x - y\|}$, with $\Bc \in \real^{|\Nc|}$ an arbitrarily large compact set. Then~\eqref{eq:sys_dyn} admits a unique equilibrium which is globally asymptotically stable.
\end{theorem}

\begin{proof}
    According to Proposition~\ref{prop:unique_equi}, there exists a unique equilibrium $(\qbf^\star,\vbf^\star)$ for~\eqref{eq:sys_dyn}. Consider the following discrete-time Lyapunov function
    \begin{align}\label{eq:Lyapunov}
        D(t) = \|\vbf(t) - \vbf^\star\|_{\Xbf^{-1}},
    \end{align}
    which measures the distance between $\vbf(t)$ and the equilibrium $\vbf^\star$. It then follows that
    \begin{align}
        &D(t+1) \notag \\
        &= \|\vbf(t+1) - \vbf^\star\|_{\Xbf^{-1}} \notag\\
        &= \| \Xbf \left( \qbf(t) + \epsilon \left(\phib(\vbf(t)) - \qbf(t) \right) \right) + \hat \vbf - \vbf^\star  \|_{\Xbf^{-1}} \notag \\
        &= \| (1-\epsilon)\Xbf\qbf(t) + \epsilon \Xbf\phib(\vbf(t)) + \hat \vbf - \vbf^\star \|_{\Xbf^{-1}} \notag \\
        &= \|(1-\epsilon)(\vbf(t) - \vbf^\star) + \epsilon(\Xbf\phib(\vbf(t)) + \hat\vbf - \vbf^\star) \|_{\Xbf^{-1}} \notag \\
        &= \|(1-\epsilon)(\vbf(t) - \vbf^\star) + \epsilon\Xbf \left(\phib(\vbf(t)) - \phib(\vbf^\star \right) \|_{\Xbf^{-1}},
    \end{align}
    where we use the fact $\vbf^\star = \Xbf\qbf^\star + \hat\vbf = \Xbf\phib(\vbf^\star) + \hat\vbf$ in the last equality. Expanding this expression,
    \begin{align}
        D(t+1) =& (1-\epsilon)^2 \| \vbf(t) - \vbf^\star \|_{\Xbf^{-1}} \notag \\
        &+ 2\epsilon(1-\epsilon) \left( \phib(\vbf(t)) - \phib(\vbf^\star)  \right)^\top \left( \vbf(t) - \vbf^\star \right) \notag \\
        &+ \epsilon^2 \| \phib(\vbf(t)) - \phib(\vbf^\star)  \|_{\Xbf}.
    \end{align}
    Therefore, leveraging the monotonicity property~\eqref{eq:monotonicity}, we have
    \begin{align}\label{eq:lyapunov_diff}
        &D(t+1) - D(t) \notag\\
        &= \epsilon(\epsilon-2) \| \vbf(t) - \vbf^\star \|_{\Xbf^{-1}} + \epsilon^2 \| \phib(\vbf(t)) - \phib(\vbf^\star )\|_{\Xbf} \notag \\
        &\quad + \underbrace{2\epsilon(1-\epsilon) \left( \phib(\vbf(t)) - \phib(\vbf^\star)  \right)^\top \left( \vbf(t) - \vbf^\star \right)}_{\leq 0} \notag\\
        &\leq \epsilon(\epsilon-2) \| \vbf(t) - \vbf^\star \|_{\Xbf^{-1}} + \epsilon^2 \| \phib(\vbf(t)) - \phib(\vbf^\star ) \|_{\Xbf}.
    \end{align}
    We show next that $D(t+1) - D(t) < 0$ for all $t \geq 0$ whenever $\vbf(t) \neq \vbf^\star$, which implies the global asymptotic stability of the equilibrium. It is straightforward to see that, for this inequality to hold, it is sufficient to ask that
    \begin{align}\label{eq:cond_epsilon_inverse}
        \epsilon < \min_{t \geq 0} \left\{ \frac{2 \| \vbf(t) - \vbf^\star \|_{\Xbf^{-1}}}{\| \vbf(t) \!-\! \vbf^\star \|_{\Xbf^{-1}} \!+\! \| \phib(\vbf(t)) \!-\! \phib(\vbf^\star ) \|_{\Xbf}} \right\} ,
    \end{align}
    which is equivalent to
    \begin{align}
        \frac{1}{\epsilon} > \max_{t \geq 0} \left\{ \frac{\| \vbf(t) \!-\! \vbf^\star \|_{\Xbf^{-1}} \!+\! \| \phib(\vbf(t)) - \phib(\vbf^\star ) \|_{\Xbf}}{2 \| \vbf(t) \!-\! \vbf^\star \|_{\Xbf^{-1}}} \right\}.
    \end{align}
    Note that $\| \phib(\vbf(t)) - \phib(\vbf^\star ) \|_{\Xbf} \leq L^2 \|\Xbf\| \|\vbf(t) - \vbf^\star \|^2$, while $\| \vbf(t) - \vbf^\star \|_{\Xbf^{-1}} \geq \|\Xbf\|^{-1} \|\vbf(t) - \vbf^\star \|^2$.
    It follows that
    \begin{align}
        &\max_{t \geq 0} \left\{ \frac{\| \vbf(t) - \vbf^\star \|_{\Xbf^{-1}} + \| \left(\phib(\vbf(t)) - \phib(\vbf^\star \right) \|_{\Xbf}}{2 \| \vbf(t) - \vbf^\star \|_{\Xbf^{-1}}} \right\} \notag\\
        & \leq \frac{1}{2} + \frac{L^2\|\Xbf\|}{2 \|\Xbf\|^{-1}} = \frac{1+L^2\|\Xbf\|^2}{2}.
    \end{align}
    Therefore, $\epsilon < \frac{2}{1+L^2\|\Xbf\|^2}$ is sufficient to make sure~\eqref{eq:cond_epsilon_inverse} hold. This, together with $\epsilon \in [0,1]$, completes the proof.
\end{proof}
Theorem~\ref{thm:aymp_stability} indicates that given $\phib$ is a monotonic function w.r.t. $\vbf$, one can always find a small enough $\epsilon$ such that the closed-loop system is globally asymptotically stable under~\eqref{eq:bus_react_upd}. The upper bound of $\epsilon$ relies on the network parameter $\| \Xbf \|$ and the Lipschitz constant $L$ for selected $\phib$, which can be easily computed offline before the online implementation of the controllers. Or one can continuously scale down $\epsilon$ during the online implementation until the controllers work well.

\begin{remark}\longthmtitle{Less conservative stability conditions compared to prior work and the role of communication in improving performance}\label{rmk:conservativeness}
Previous works~\cite{GC-ZY-MKS-JC:22-cdc,ZY-GC-MKS-JC:24-tps,JF-YS-GQ-SHL-AA-AW:24,JF-WC-JC-YS:23-csl} have identified properties on decentralized voltage controllers to ensure closed-loop stability. Considering the same incremental update rule as~\eqref{eq:bus_react_upd}, the proof techniques in these works would amount to the following inequality for each bus to ensure closed-loop stability in the context of this work:
    \begin{align}\label{eq:prior_cond}
        \left( \phi_i(v_i, \{v_j\}_{j \in \Nc_i}) - \phi_i(v_i^\prime, \{v_j^\prime\}_{j \in \Nc_i}) \right) (v_i - v_i^\prime) \leq 0,
    \end{align}
    for all $i \in \Nc$ and $\vbf,\vbf^\prime \in \real^N$, i.e., a monotonicity condition for decentralized controller at each bus $i \in \Nc$. 
    These conditions imply~\eqref{eq:monotonicity}, but not the other way around. Thus, the condition obtained in this paper is more general.  
    
    In fact, for our proposed framework, the condition is relaxed such that the controllers are only required to collectively satisfy the monotonicity condition. For instance, condition~\eqref{eq:monotonicity} is satisfied if 
    \begin{subequations}\label{eq:new_cond_budget}
    \begin{align}
        \left( \phi_i(v_i, \{v_j\}_{j \in \Nc_i}) - \phi_i(v_i^\prime, \{v_j^\prime\}_{j \in \Nc_i}) \right) (v_i - v_i^\prime) &\leq b_i, \label{eq:new_cond_budget_a}\\
        \sum_{i \in \Nc} b_i &= 0. \label{eq:new_cond_budget_b}
    \end{align}
    \end{subequations}
    This allows some buses to violate~\eqref{eq:prior_cond} up to a level that can be compensated by other buses to still make the overall summation~\eqref{eq:monotonicity-separate} non-positive. We can think of $\{b_i\}_{i \in \Nc}$ as \emph{budgets}, cf.~\cite{ZY-CZ-JC:24-scl}, that allow us to make the stability condition less conservative. The condition~\eqref{eq:prior_cond} turns out to be a special case of~\eqref{eq:new_cond_budget} with $\{b_i\}_{i \in \Nc}$ all zero. 
    
    This comparison also highlights the value of the role of communication:~\eqref{eq:new_cond_budget_b} can be enforced via coordination among buses with non-zero $b_i$'s, providing more flexibility to improve the overall performance.
    Another important observation is that, unlike most of the existing voltage controller designs, e.g.,~\cite{KT-PS-SB-MC:11,IEEE1547,HZ-HJL:15}, our framework here allows the controller at each bus $i \in \Nc$ to not be necessarily monotonic w.r.t. its local voltage magnitude, while still guaranteeing closed-loop stability. We envision this is particularly meaningful in the design of provably stable voltage controllers for DGs with relatively small capable DERs, and in scenarios where the system operator aims on minimizing power losses instead of voltage deviations, as in such cases the monotonicity requirement on every bus might be the main source of conservativeness~\cite{ZY-GC-JC:23-csl}.
    \oprocend
\end{remark}

Our discussion so far has shown the uniqueness of equilibrium and its global asymptotic stability assuming no uncertainties are present. However, in practice, the system dynamics might be subject to different sources of error, including disturbances, noise, and communication latency. To account for these, we instead consider the dynamics
\begin{subequations}\label{eq:sys_dyn_disturbed}
  \begin{align}
    \qbf(t+1) &= \qbf(t) + \epsilon \left(\phib(\vbf(t)) - \qbf(t) \right) + \dbf_{\qbf}, \label{eq:sys_dyn_disturbed_q}  \\
    \vbf(t+1) &= \Xbf \qbf(t+1) + \hat \vbf + \dbf_{\vbf}, \label{eq:sys_dyn_disturbed_v} 
\end{align}  
\end{subequations}
where $\dbf_\qbf$ and $\dbf_\vbf$ encode the error terms. The following result establishes the input-to-state stability (ISS) of~\eqref{eq:sys_dyn_disturbed}. 

\begin{proposition}\longthmtitle{Input-to-state stability (ISS) of the closed-loop system~\eqref{eq:sys_dyn_disturbed}}\label{prop:ISS}
Let $\phib$ be monotonically decreasing w.r.t. $\vbf$, cf.~\eqref{eq:monotonicity}, and $\epsilon$ satisfy the condition~\eqref{eq:cond_epsilon}. Then, the closed-loop system~\eqref{eq:sys_dyn_disturbed} is input-to-state stable.
\end{proposition}
\begin{proof}
Consider the closed-loop system~\eqref{eq:sys_dyn_disturbed}, together with the Lyapunov function~\eqref{eq:Lyapunov} and condition~\eqref{eq:cond_epsilon}. In view of~\eqref{eq:lyapunov_diff}, it follows that
    \begin{align}
        &D(t+1) - D(t) \notag\\
        &\leq \epsilon(\epsilon - 2) \| \vbf(t) - \vbf^\star \|_{\Xbf^{-1}} + \epsilon^2 \| \phib(\vbf(t)) - \phib(\vbf^\star )\|_{\Xbf} \notag \\
        & \qquad+ \| \dbf_\qbf\| + \| \dbf_\vbf\|_{\Xbf^{-1}} \notag \\
        &\leq \epsilon((1+ L^2\|\Xbf\|^2)\epsilon-2) \| \vbf(t) - \vbf^\star \|_{\Xbf^{-1}} + \| \dbf_\qbf\|  + \| \dbf_\vbf\|_{\Xbf^{-1}} \notag \\
        &= \underbrace{\epsilon((1+ L^2\|\Xbf\|^2)\epsilon-2)}_{< 0} D(t) + \| \dbf_\qbf\|  + \| \dbf_\vbf\|_{\Xbf^{-1}},
    \end{align}
where we have used the fact that $\| \phib(\vbf(t)) - \phib(\vbf^\star ) \|_{\Xbf} \leq L^2 \|\Xbf\| \|\vbf(t) - \vbf^\star \|^2 = L^2 \|\Xbf\|^2 \|\Xbf\|^{-1} \|\vbf(t) - \vbf^\star \|^2 \leq L^2 \|\Xbf\|^2 \| \vbf(t) - \vbf^\star \|_{\Xbf^{-1}}$ in the second inequality. Combining with the fact that $\underline \lambda \|\vbf - \vbf^\star\|^2 \leq D \leq \overline \lambda \|\vbf - \vbf^\star\|^2$, where $\underline \lambda$ and $\overline \lambda$ are respectively the minimum and maximum eigenvalues of $\Xbf^{-1}$, we conclude that $D(t)$ is an ISS-Lyapunov function and thus, by~\cite[Lemma 3.5]{ZPJ-YW:01},
the system is input-to-state stable w.r.t. $[\dbf_\qbf;\dbf_\vbf]$.
\end{proof}

The ISS property established in Proposition~\ref{prop:ISS} is a characterization of the robustness properties of the designed controllers~\eqref{eq:bus_react_upd}. For instance, communication imperfections or noise in reactive power and voltage magnitude measurements give rise to errors modeled by $\dbf_\qbf$, whereas unmodeled dynamics or the actual nonlinearity of the power flow model can be encoded by $\dbf_\vbf$. The ISS property implies the graceful degradation of the dynamical behavior of the system~\eqref{eq:sys_dyn_disturbed}, with the deviation from the nominal behavior being an increasing function of the error magnitude.


\subsection{The Case When Only a Subset of Buses Are Controllable}
The framework we introduce above assumes that all the buses are equipped with DERs that are able to perform reactive power control. However, in some scenarios, there might exist buses which are load buses and have no voltage control capability. We show here how these scenarios can also be covered by our design: as long as the controllable buses satisfy the monotonicity requirement~\eqref{eq:monotonicity-separate}, the stability condition~\eqref{eq:monotonicity} holds automatically.
To be more specific, let $\Cc \subseteq \Nc$ and $\Uc := \Nc \setminus \Cc$ denote the sets of controllable buses and uncontrollable buses, respectively. It is easy to see that the power flow equation~\eqref{eq:v=Rp+Xq} can be partitioned as
\begin{align}
    \begin{bmatrix}
\vbf_\Uc \\
\vbf_\Cc
\end{bmatrix} = \begin{bmatrix}
\Rbf_{\Uc\Uc} & \Rbf_{\Uc\Cc} \\
\Rbf_{\Uc\Cc}^\top & \Rbf_{\Cc\Cc} 
\end{bmatrix} \begin{bmatrix}
\pbf_\Uc \\
\pbf_\Cc
\end{bmatrix} + \begin{bmatrix}
\Xbf_{\Uc\Uc} & \Xbf_{\Uc\Cc} \\
\Xbf_{\Uc\Cc}^\top & \Xbf_{\Cc\Cc} 
\end{bmatrix} \begin{bmatrix}
\qbf_\Uc \\
\qbf_\Cc
\end{bmatrix} + \ones,
\end{align}
where $\Rbf_{\Uc\Uc},\Rbf_{\Cc\Cc},\Xbf_{\Uc\Uc},\Xbf_{\Cc\Cc} \succ 0$~\cite{RAH-CRJ:12} and it follows that
\begin{align}\label{eq:pf_subset}
    \vbf_\Cc = \Xbf_{\Cc\Cc} \qbf_\Cc + \tilde{\vbf},
\end{align}
in which $\tilde{\vbf} = \Rbf_{\Uc\Cc}^\top \pbf_\Uc + \Rbf_{\Cc\Cc} \pbf_\Cc + \Xbf_{\Uc\Cc}^\top \qbf_\Uc + \ones$. Therefore, one can derive the closed-loop system dynamics in the same form as~\eqref{eq:sys_dyn}, and thus our theoretical results keep valid.

\section{Learning Stable Voltage Controllers under Arbitrary Communication Infrastructures}\label{sec:learning}
Given the identified constraints on $\{\phi_i\}_{i \in \Nc}$ to ensure the closed-loop system stability, here we are interested in developing a data-driven framework for synthesizing optimal $\{\phi_i\}_{i \in \Nc}$ which satisfy the stability constraints by design given an arbitrary communication infrastructure. This gives rise to the following optimization problem:
\begin{subequations}\label{eq:ORPF}
	\begin{align}
	\min_{\{\phi_i\}_{i \in \Nc}}\ &  ~ \sum_{k=1}^{K}  \sum_{t=0}^{T-1} f(\pbf^k,\qbf(t)) \label{eq:ORPF:cost}\\
	\mathrm{s.t.}\  & ~ \textrm{Power Flow Model}~\eqref{eq:nonlinear_pf} \label{eq:ORPF:pf}\\
                        & ~ \qbf(t+1) = \qbf(t) + \epsilon \left(\phib(\vbf(t)) - \qbf(t) \right) \label{eq:ORPF:controller} \\
                        & ~ \phib~\textrm{is monotone w.r.t.}~\vbf \label{eq:ORPF:stability} \\
                        & ~ \phib \in \Qc := \times_{i \in \Nc} \Qc_i \label{eq:ORPF:reactive_pwr_bound}
	\end{align} 
\end{subequations}
where we aim to optimize the cost function of interest $f: \real^{2N} \mapsto \real$ in~\eqref{eq:ORPF:cost} for a total of $K$ load-generation scenarios, with $T$ the transient time horizon for each scenario. 
A common choice of the cost function $f$ is the penalization on the voltage deviation from desired limits, i.e., $\|\vbf(\pbf^k,\qbf(t)) - \ones\|$. Other formulations may consider electric losses in the network or the deviation from a pre-determined substation power trajectory. The constraint~\eqref{eq:ORPF:reactive_pwr_bound} defines the capability of reactive power compensation, while the constraints in~\eqref{eq:ORPF:stability} ensure that the designed controllers render the closed-loop system globally asymptotically stable to a unique equilibrium, cf.~Theorem~\ref{thm:aymp_stability}.
Note that we do not explicitly include voltage constraints in~\eqref{eq:ORPF} 
because we do not make any assumptions on the load/generation scenario and there are reactive power capability constraints~\eqref{eq:ORPF:reactive_pwr_bound}.
\begin{figure*}[htb]
    \centering
    \includegraphics[width=0.95\textwidth]{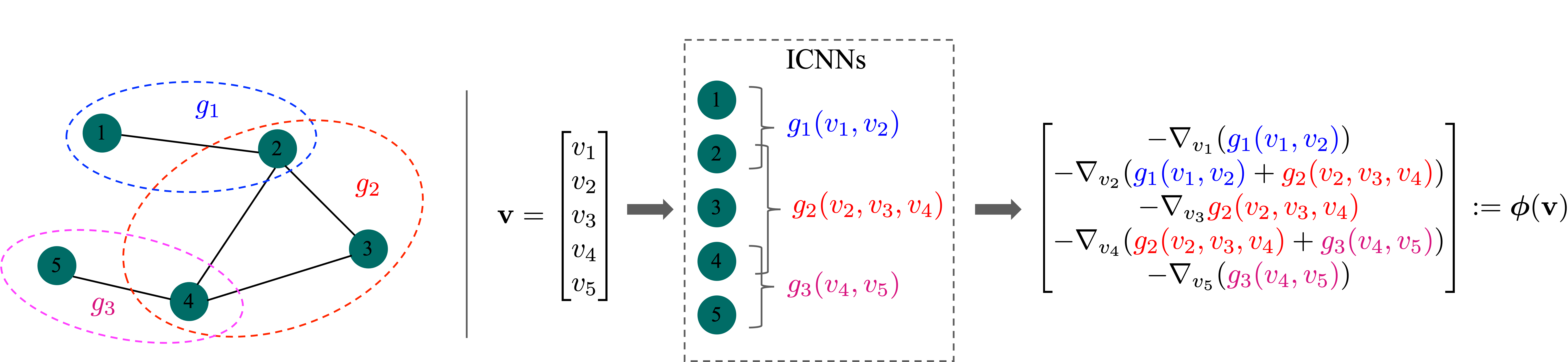}
    \caption{Illustration of the proposed design procedure to compute equilibrium functions under the distributed communication infrastructure example in Fig.~\ref{fig:communication_setups_distributed}.}
    \label{fig:monotone_design_distributed}
\end{figure*}
The optimization problem~\eqref{eq:ORPF} is not amenable to relevant learning algorithms because of the constraint~\eqref{eq:ORPF:stability}, which restricts the behavior of the control policy to be learned. 

Similar to our prior works~\cite{JF-YS-GQ-SHL-AA-AW:24,ZY-GC-MKS-JC:24-tps}, the idea to solve~\eqref{eq:ORPF} is to construct structured neural networks for parameterizing $\{\phi_i\}_{i \in \Nc}$ such that~\eqref{eq:ORPF:stability} is satisfied by design. Unlike the decentralized case in~\cite{JF-YS-GQ-SHL-AA-AW:24,ZY-GC-MKS-JC:24-tps}, where each $\{\phi_i\}_{i \in \Nc}$ is simply paramterized by a single-input single-output monotone Rectified Linear Unit (ReLU) neural network, here we focus on the case that there is a communication infrastructure that allows $\{\phi_i\}_{i \in \Nc}$ to coordinate with each other to satisfy the less-conservative \emph{collective monotone} requirement~\eqref{eq:monotonicity}. In the next, we draw inspirations from~\cite{WC-YJ-BZ-YS:23} and propose the following unified design procedure to strictly achieve this by design under arbitrary communication infrastructure leveraging input convex neural networks (ICNNs)~\cite{BA-LX-JZK:17}, which involves the following steps.
\subsubsection*{Step 1: Communication network partition} Partition the communication graph into subgraphs such that the buses in each subgraph are all-to-all connected. Note that one bus may appear in multiple subgraphs.
\subsubsection*{Step 2: ICNN for each subgraph} Suppose there are in total $S$ subgraphs indexed by $\Sc \triangleq \{1,...,S\}$ and the $\ell$-th subgraph contains buses $\Mc_\ell \subseteq \Nc$. Design ICNNs $g_\ell: \real^{|\Mc_\ell|} \rightarrow \real$ for all $\ell \in \Sc$, where the inputs are voltage magnitude measurements from the buses in the $\ell$-th subgraph.
\subsubsection*{Step 3: Derivation of $\{\phi_i\}_{i \in \Nc}$ through computing gradients of ICNNs} For each $i \in \Nc$, let $\phib(\vbf) := - \sum_{\ell \in \Sc} \nabla_{\vbf} g_\ell(\vbf)$.

We emphasize that this procedure is naturally suitable for different communication scenarios, as every bus only uses the information available from itself as well as its communicating neighbors to determine its equilibrium function. To be more illustrative, we show in Fig.~\ref{fig:monotone_design_distributed} how this procedure works for the distributed communication infrastructure displayed in Fig.~\ref{fig:communication_setups_distributed}. We note that the design procedure works similarly for other arbitrary communication infrastructures. In the next result, we provide a formal result that verifies this design procedure yields equilibrium functions that satisfy~\eqref{eq:monotonicity}.

\begin{proposition}\longthmtitle{Unified equilibrium function design procedure}\label{prop:design_procedure}
Given that all the functions $\{g_\ell\}_{\ell \in \Sc}$ are convex, the obtained equilibrium functions $\{\phi_i\}_{i \in \Nc}$ are guaranteed to satisfy~\eqref{eq:monotonicity} under arbitrary communication infrastructure.   
\end{proposition}
\begin{proof}
    It follows that
    \begin{align}
        &\left( \phib(\vbf) - \phib(\vbf^\prime)  \right)^\top \left( \vbf - \vbf^\prime \right) \notag\\
        &= - \Big(\sum_{\ell \in \Sc} \nabla_{\vbf} g_\ell(\vbf) - \sum_{\ell \in \Sc} \nabla_{\vbf} g_\ell(\vbf^\prime) \Big) \left( \vbf - \vbf^\prime \right) \notag\\
        &= - \sum_{\ell \in \Sc} \left(\nabla_{\vbf} g_\ell(\vbf) - \nabla_{\vbf} g_\ell(\vbf^\prime) \right) \left( \vbf - \vbf^\prime \right).
    \end{align}
    Since $g_\ell$ is convex for all $\ell \in \Sc$, it holds that $- \left(\nabla_{\vbf} g_\ell(\vbf) - \nabla_{\vbf} g_\ell(\vbf^\prime) \right)^\top \left( \vbf - \vbf^\prime \right) \leq 0$, we conclude that~\eqref{eq:monotonicity} is satisfied.
\end{proof}
We note that, since the ICNN has universal approximation capability for any convex function~\cite{WC-YJ-BZ-YS:23}, the above design does not pose any limitations on the expressiveness of $\phib$. Given Proposition~\ref{prop:design_procedure}, we are now able to rewrite the optimization problem~\eqref{eq:ORPF} into the following learning-amenable form\footnote{We add the projection operator in~\eqref{eq:ORPF-new:q} to comply with the operational constraints in practical implementation, which is not considered in the theoretical analysis.}:
\begin{subequations}\label{eq:ORPF-new}
	\begin{align}
	\min_{\{\alpha_\ell,g_\ell\}_{\ell \in \Sc}}\ &  ~ \sum_{k=1}^{K}  \sum_{t=0}^{T-1} f(\pbf^k,\qbf(t)) \label{eq:ORPF-new:cost}\\
	\mathrm{s.t.}\  & ~ \textrm{Power Flow Model~\eqref{eq:nonlinear_pf}} \\
                        & ~ \qbf(t+1) = \qbf(t) + \epsilon \left(\phib(\vbf(t)) - \qbf(t) \right) \label{eq:ORPF-new:controller}\\
                        & ~\phi_i(\vbf) = \proj_{\Qc_i} (- \nabla_{v_i} \sum_{\ell \in \Sc} g_\ell(\vbf)), \forall i \in \Nc \label{eq:ORPF-new:q} \\
                        & ~\{g_\ell\}_{\ell \in \Sc}~\textrm{are convex}.
	\end{align} 
\end{subequations}
By parameterizing $\{g_\ell\}_{\ell \in \Sc}$ using ICNNs, problem~\eqref{eq:ORPF-new} can be tractably solved using relevant machine learning algorithms. 

\section{Simulations}
Here, we evaluate our approach on the reduced
UCSD microgrid testbed~\cite{BW-JD-DW-JK-NB-WT-CR:13} with different levels of communication. Though our theoretical analysis is based on the linearized power flow model, all experiments in this section are run using the nonlinear power flow simulator in Pandapower~\cite{LT-AS-FS-JHM-JD-FM-SM-MB:18} to evaluate the algorithm performance.

\subsection{UCSD Microgrid Testbed and the Dataset}
The network model of the reduced UCSD microgrid testbed is shown in Fig.~\ref{fig:ucsd_microgrid},
showing the locations of the load and PV generators. This model consists of a total of 49 buses (including a substation bus), with 13 of them equipped with PV generators participating in voltage regulation. These are located at buses $\Cc = \{14,15,17,19,20,27,29,30,32,34,38,39,41\}$ and its nominal voltage magnitude is $12.47$ kV. 
We provide in Table~\ref{tab:microgrid_resistence_reactance} (in Appendix) the network's equivalent resistances and reactances. 
\begin{figure*}[t]
    \centering
    \includegraphics[width=0.75\textwidth]{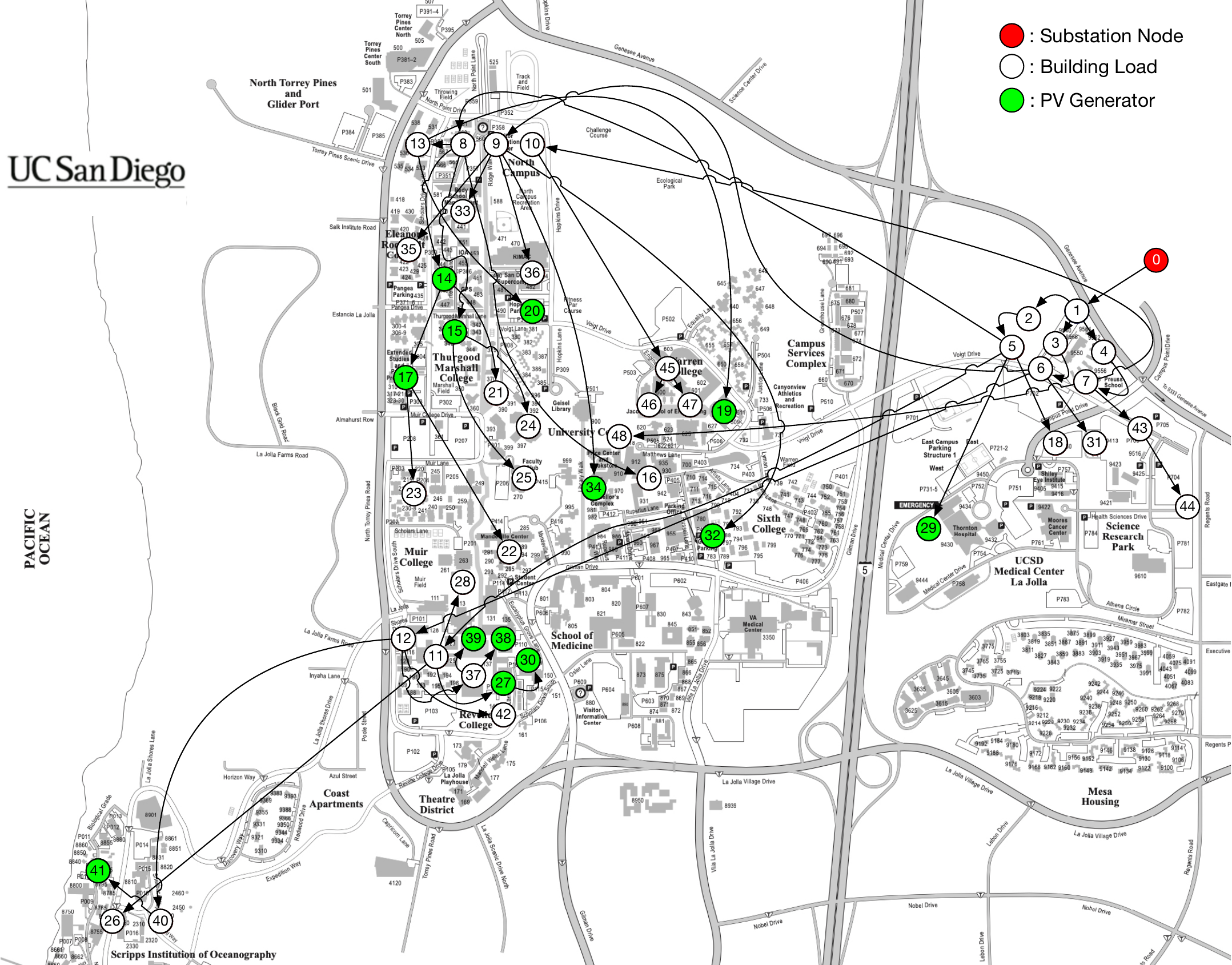}
    \caption{Reduced-order model of the UCSD microgrid testbed with 13 PV generators equipped with voltage controllers located at nodes $\Cc = \{14,15,17,19,20,27,29,30,32,34,38,39,41\}$. }
    \label{fig:ucsd_microgrid}
\end{figure*}
\begin{figure}[htb]
    \centering
    \includegraphics[width=0.8\linewidth]{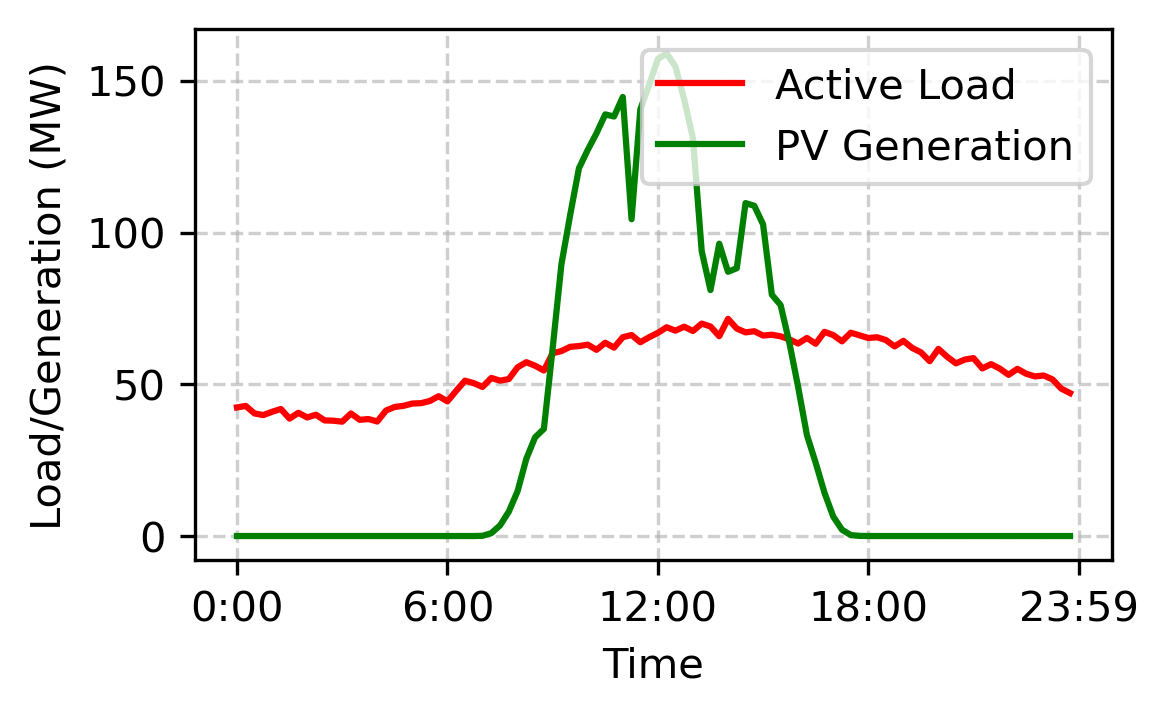}
    \caption{The load and generation curves across the UCSD microgrid on February 2, 2020 (Sunday).}
    \label{fig:load_gen}
    \vspace*{-2ex}
\end{figure}
We obtain the load and generation profiles of the UCSD microgrid by aggregating and scaling the real-gathered load and PV generation data at the UCSD campus~\cite{SS-CM-YAC-AG-JD-JK:21} in Feb 2020. We normalize the aggregated building load profiles, and then obtain the active load profiles for each bus by scaling the normalized load profiles so that the demand peak across the network is about 1.8 times of the nominal demand (42 MW~\cite{BW-JD-DW-JK-NB-WT-CR:13}). 
We synthesize the reactive load profiles according to the active load profiles and the power factors of the UCSD microgrid.
Similarly, we obtain the active generation profiles for each PV generator by scaling the 
aggregated PV generation data to induce significant voltage deviations in high voltage scenarios ($> 5\%$).
Fig.~\ref{fig:load_gen} illustrates an instance of the obtained load and generation curves across the network for one day (February 2, 2020).
Each day consists of 96 data points, i.e., the data sampling period is 15 minutes. For training purpose, we generate three times more data for each day by adding $10\%$ random noise to the 96 load and generation data points.

\subsection{Controller and Training Setups}

The controller setups with different communication levels are presented as follows:
\begin{itemize}[leftmargin=*]
    \item \underline{Controllers with No Communication}: There is no communication, see  Fig.~\ref{fig:communication_setups_decentralized}. The communication network is partitioned into $S=13$ communication subgraphs, with each controller only receives its local information. Each $\{g_\ell\}_{\ell \in \{1,..,13\}}$ is parameterized by an ICNN with two hidden layers, and each has $64$ hidden units;
    \item \underline{Controllers with Distributed Communications}: We consider two cases. For Case 1, each PV generator has a relatively small communication radius. The distributed communication network is shown in Fig.~\ref{fig:communication_setups_distributed_1}. The communication network is partitioned into $S=5$ connected subgraphs, represented as $\Gc_1=\{14, 15, 17, 20\}$, $\Gc_2=\{19, 32, 34\}$, $\Gc_3 = \{27, 30, 38, 39\}$, $\Gc_4 = \{29\}$, and $\Gc_5 = \{41\}$. 
    For Case 2, each PV generator has relatively larger communication radius, and the communication network is shown in Fig.~\ref{fig:communication_setups_distributed_2}.
    The communication network is partitioned into $S=2$ connected subgraphs, represented as $\Gc_1=\{14,15,17,19,20,27,29,30,32,34,38,39\}$, $\Gc_2=\{27, 30, 38, 39, 41\}$. Each $\{g_\ell\}_{\ell \in \{1,..,5\}}$ in the first cases and each $\{g_\ell\}_{\ell \in \{1,2\}}$ are all parameterized as an ICNN with two hidden layers, and each has $64$ hidden units; 
    \item \underline{Controllers with Full Communications}: The communication network is all-to-all connected, see Fig.~\ref{fig:communication_setups_centralized}. The communication network is partitioned into $S=1$ all-to-all connected subgraph, which contains all the 13 controlled buses. The only $g_1$ is parameterized using an ICNN with two hidden layers, and each has $64$ hidden units.
\end{itemize}
For brevity, we use \texttt{Ctrl-NC}, \texttt{Ctrl-DC-1}, \texttt{Ctrl-DC-2}, and \texttt{Ctrl-FC} to represent the controllers under no communication, distributed communication - case 1, distributed communication - case 2, and full communication scenarios. We note that the overall communication level gradually increases for these four different communication scenarios\footnote{We note that both \texttt{Ctrl-DC-2} and \texttt{Ctrl-FC} have connected communication graphs. From a classical distributed optimization viewpoint, this might make them equivalent. However, here, the controllers are making one-shot decisions based on the information received from one round of communication, instead of recursively finding a solution through repeated rounds. Therefore, \texttt{Ctrl-FC} has a richer communication level, since each controller has access to more information when compared to the case with \texttt{Ctrl-DC-2}.}, that is, \texttt{Ctrl-NC} $<$ \texttt{Ctrl-DC-1} $<$ \texttt{Ctrl-DC-2} $<$ \texttt{Ctrl-FC}. Finally, we remark that our prior works~\cite{JF-YS-GQ-SHL-AA-AW:24,ZY-GC-MKS-JC:24-tps} belong to \texttt{Ctrl-NC}, where the controllers are decentralized without any communication.

\begin{figure}[tb]
    \centering
    \subfigure[Distributed communication network - Case 1]{
        \includegraphics[scale=0.35]{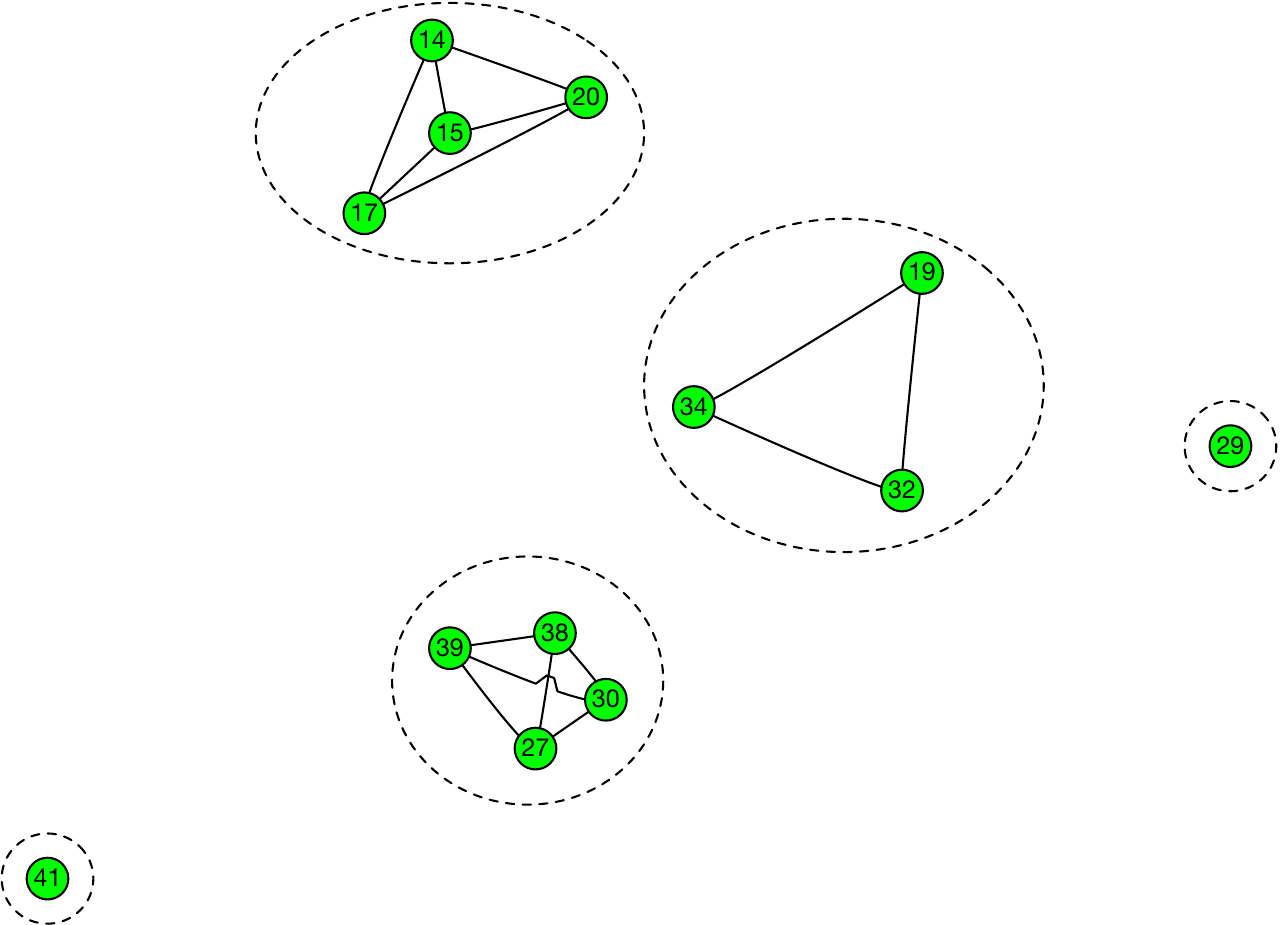}
        \label{fig:communication_setups_distributed_1}
        }
    \subfigure[Distributed communication network - Case 2]{
        \includegraphics[scale=0.35]{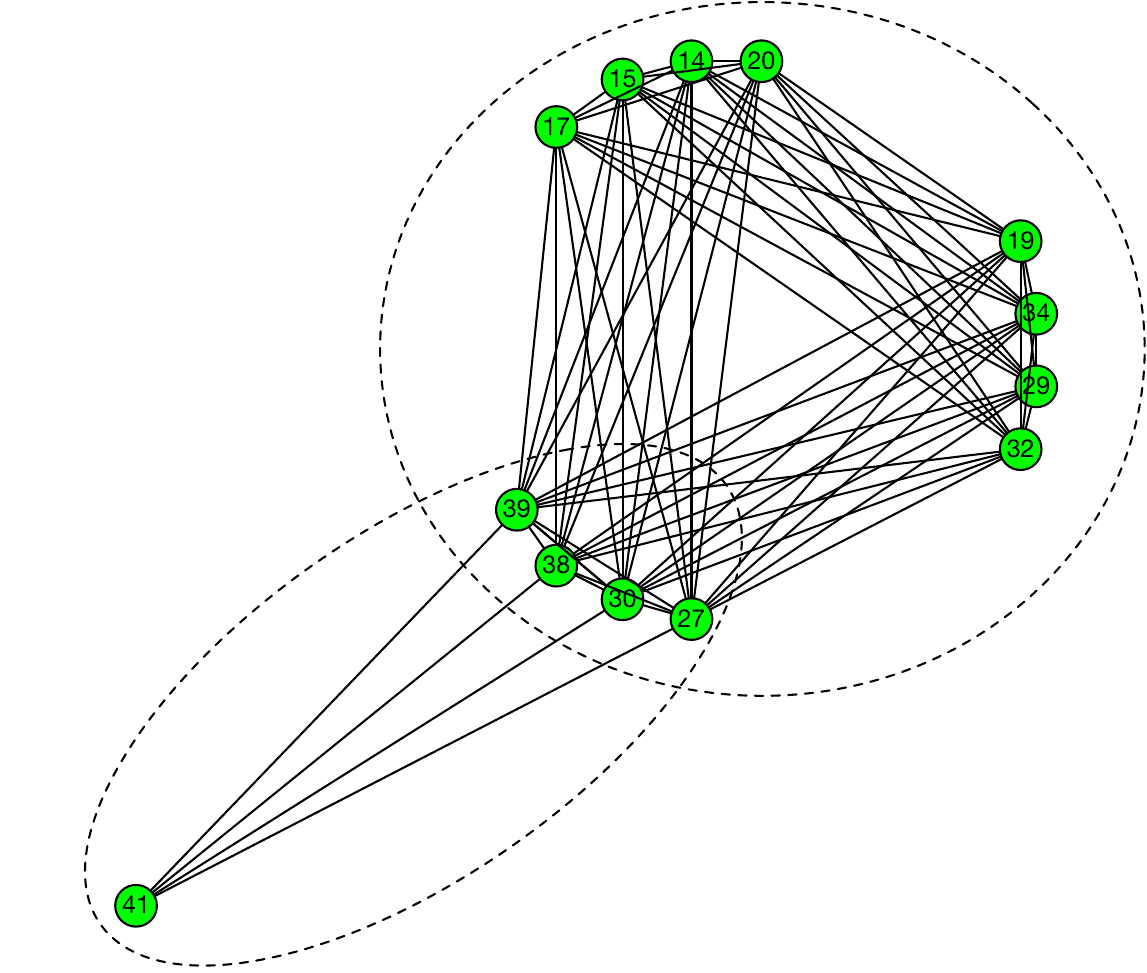}
        \label{fig:communication_setups_distributed_2}
        }
  \caption{Two different distributed Communication networks (we have slightly modified the relative physical location between buses to make the connectivity clearer). For Case 1/Case 2, the networks can be divided into five/two subgraphs (circled by dashed lines) which all have all-to-all communication.}
  \vspace*{-2ex}
\end{figure}

We use supervised learning to train the controllers parameterized by ICNNs. The goal is to make $\phib$ approximate the OPF solutions at every step in~\eqref{eq:ORPF-new}. To achieve this, for each load and generation profile $\{\pbf^k,\qbf_\Uc^k\}_{k=1}^K$ in the dataset, we respectively compute the voltage magnitudes $\vbf^{k}_\Cc$ and the optimal reactive power setpoints $\qbf^{\star,k}_\Cc$ for the PV generators in $\Cc$ by solving the following OPF problem:
\begin{align}
    \min_{\qbf}\ &  ~ f(\pbf^k,\qbf_\Uc^k,\qbf_\Cc) \\
    \mathrm{s.t.}\  & ~ \textrm{Power Flow Model~\eqref{eq:nonlinear_pf}}, \qbf \in \Qc. \notag
\end{align} 
We choose the cost function $f(\pbf,\qbf_\Uc,\qbf)$ as:
\begin{align}
    100\| \vbf(\pbf,\qbf_\Uc,\qbf_\Cc) - \ones \|^2 + \begin{bmatrix}
\qbf_\Uc^\top \ \qbf_\Cc^\top
\end{bmatrix}^\top \Rbf \begin{bmatrix}
\qbf_\Uc \\
\qbf_\Cc
\end{bmatrix},
\end{align}
where the first term is to penalize the voltage violations while the second term is to minimize the power losses induced by reactive power injections~\cite{GC-RC:17}. We refer to the first term as the \emph{voltage deviation cost} and the second term \emph{power loss cost}. We multiply by a factor $100$ the first term to balance the values between these two cost terms. We set the reactive power capability constraint $\Qc$ to be $\{\qbf_\Cc: -\qbf_{\rm lim} \leq \qbf_\Cc \leq \qbf_{\rm lim} \}$, where $\qbf_{\rm lim} = [2,2,2,2,2,5,2,5,5,5,5,5,5]^\top$ (MVar). 

By collecting $\vbf^{k}_\Cc$ and $\qbf^{\star,k}_\Cc$ for all $k \in \{1,...,K\}$ into the labeled dataset, the supervised learning can be easily realized by solving $\min_{\phib} \sum_{k=1}^K \| \qbf^{\star,k}_\Cc - \phib(\vbf^{k}_\Cc) \|$ using supervised learning algorithms
under different communication scenarios in an offline and centralized fashion.

\subsection{Simulation Results}\label{sec:simu_results}
We train the controllers using the load and generation data of the first 22 days in the dataset. To begin with, we analyze the trained controllers at buses 27 and 29 to gain an intuitive understanding of its behavior, as shown in Figure~\ref{fig:ucsd_controllers}. 
The $y$-axis shows the outputs $\phi_{27}(\vbf_\Cc^k)$ and $\phi_{29}(\vbf_\Cc^k)$ under different communication scenarios for all $\{\vbf_\Cc^k\}_{k=1}^K$ in the labeled dataset as well as the OPF solutions.
The OPF solutions, influenced by all buses across the network, demonstrates highly versatile behavior at a selected bus. This is called the \emph{data inconsistency} phenomenon~\cite{ZY-GC-MKS-JC:24-tps}, as a fixed voltage magnitude at a single bus could have multiple different optimal control actions due to different status of other buses, and the optimal control actions do not form an significant monotonically decreasing shape w.r.t. voltage magnitude. Consequently, the \texttt{Ctrl-FC}, with global observation, closely aligns with the OPF solutions, while the \texttt{Ctrl-NC}, relying solely on local information, captures only the general trend. While for the \texttt{Ctrl-DC-1} and \texttt{Ctrl-DC-2}, with partial communication, achieve prediction accuracy between \texttt{Ctrl-FC} and \texttt{Ctrl-NC}. We also note that \texttt{Ctrl-DC-2} predicts better than \texttt{Ctrl-DC-1} as the former has a higher communication level compared to the latter. This also validates our statement in Remark~\ref{rmk:conservativeness} that the control action for each bus is forced to be monotonically decreasing w.r.t. local voltage measurement for the \texttt{Ctrl-NC} to ensure closed-loop stability, and the inclusion of communication can break it to certain level to make the control action more flexible and accurate.
\begin{figure}[t]
    \centering
    \subfigure[Control actions at PV generator 27]{
    \includegraphics[width=0.9\linewidth]{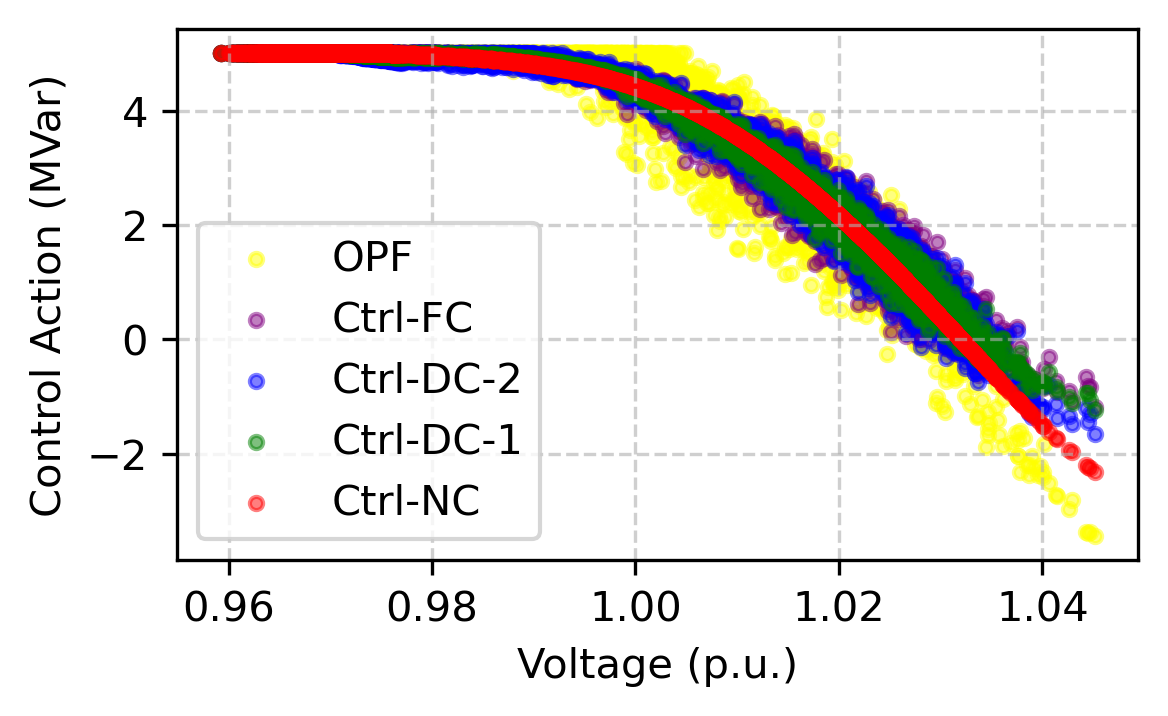}
    }
    \subfigure[Control actions at PV generator 29]{
    \includegraphics[width=0.9\linewidth]{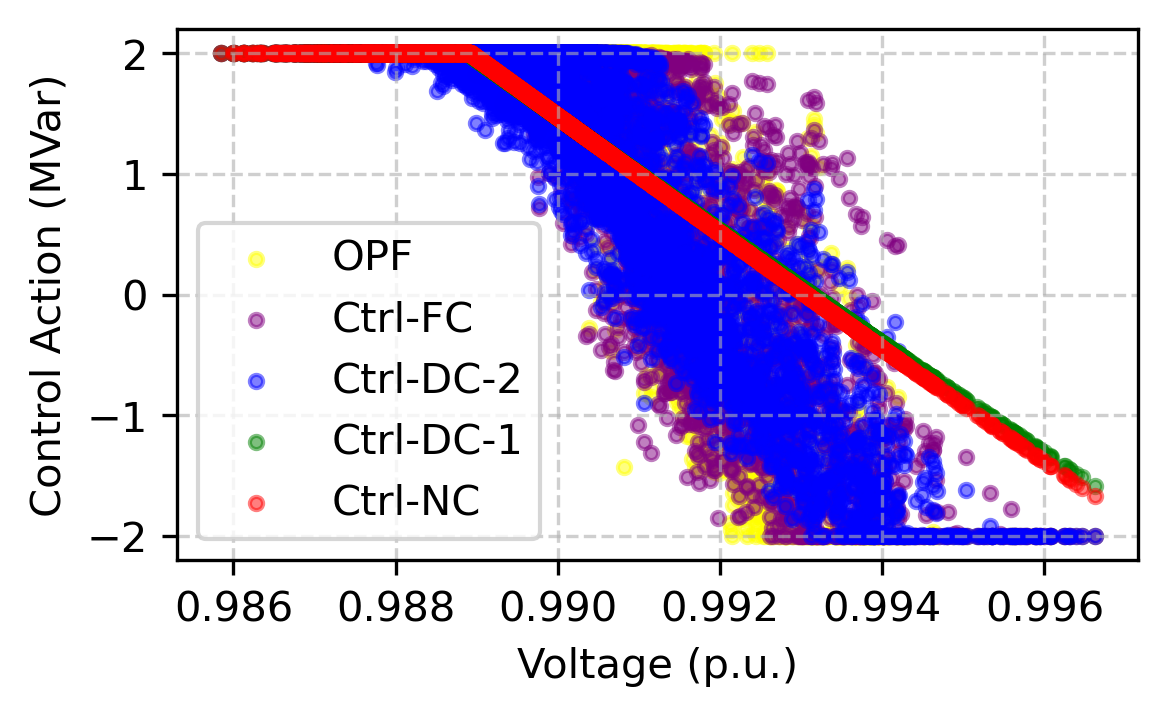}
    }
    \caption{Control actions of the trained controllers at (a) bus 27 and (b) bus 29 under controllers with different communication scenarios.}
    \label{fig:ucsd_controllers}
\end{figure}

\begin{figure*}[t]
    \centering
    \includegraphics[width=0.9\linewidth]{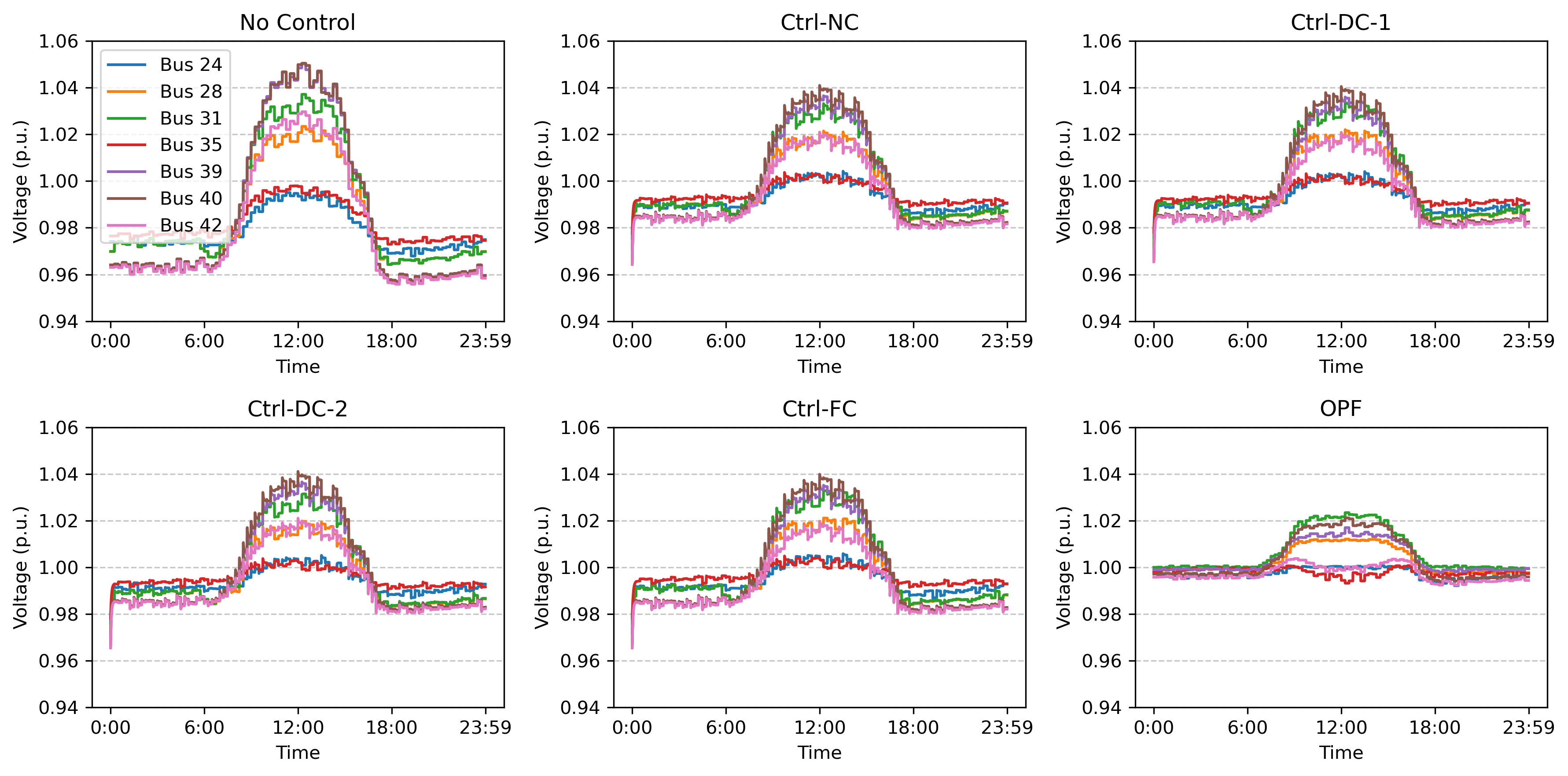}
    \caption{Voltage trajectories for selected buses (buses 24, 28, 31, 35, 39, 40, 42) under the \texttt{No Ctrl}, \texttt{Ctrl-NC}, \texttt{Ctrl-DC-1}, \texttt{Ctrl-DC-2}, \texttt{Ctrl-FC}, as well as \texttt{OPF} using the load and generation data on February 29, 2020 with random noise.}
    \label{fig:voltage_traj}
\end{figure*}

Next, we test the trained controllers in the UCSD microgrid using the load and generation data on February 29, 2020 (unseen during training) with $10\%$ random noise. Specifically, we consider $T=30$, i.e., for each data point, the reactive power iterations~\eqref{eq:bus_react_upd} run for 30 times, and the stepsize $\epsilon$ in controller \eqref{eq:bus_react_upd} is selected to be $0.1$. To show the effectiveness of the proposed framework, we compare in Fig.~\ref{fig:voltage_traj} the voltage trajectories for selected buses under the \texttt{Ctrl-DC-1} as well as \texttt{No Ctrl} and \texttt{OPF}, which represent the no control case and OPF case, respectively. 

It can be observed that compared to the \texttt{No Ctrl} case, the voltage deviations under \texttt{Ctrl-DC-1} are effectively mitigated, especially for low-voltage scenarios. 
To quantitatively compare the control performance under different communication scenarios, we provide the daily accumulated cost for different controllers in Table~\ref{table:voltage_control_ucsd}. Since all controllers are trained using OPF solutions through supervised learning, they all achieve significant improvement compared to the \texttt{No Ctrl} case. Notably, as communication level increases, the overall performance steadily improves, and the behavior of the controller gradually become closer to the \texttt{OPF} solutions, in the sense that it tends to slightly increase the power loss cost for more significant reduction of voltage deviation cost. 
\begin{table}[tb]
\centering
\caption{Accumulated cost along the day on UCSD Microgrid for controllers under different communication scenarios. }
 \label{table:voltage_control_ucsd}
 \begin{tabular}{ccccc}
    \toprule
    {Controller} & {Cost-Volt} & {Cost-Loss} & {Total Cost} & Improvement\\
    \midrule
    \texttt{No Ctrl} & 2.0614 & 0.3143 & 2.3757 & - \\ 
    \texttt{Ctrl-NC} & 0.6148 & 0.1571  & 0.7719 & 67.5\% \\ 
    \texttt{Ctrl-DC-1} & 0.6001 & 0.1574 & 0.7575 & 68.1\% \\  
    \texttt{Ctrl-DC-2} & 0.5588 & 0.1623 & 0.7211 & 69.6\% \\
    \texttt{Ctrl-FC} & 0.5334 & 0.1653 & 0.6987 & 70.6\%  \\
    \texttt{OPF} & 0.1883 & 0.2673 & 0.4556 &  80.8\% \\
    \bottomrule
\end{tabular}
\end{table}

Finally, we test the robustness of the proposed framework. We use the same  controllers learned above but consider the case where there exists random voltage magnitude measurement noise. 
Table~\ref{tab:cost_noise} summarizes the averaged cost under different levels of voltage magnitude measurement noise.
Specifically, 0.5\% noise in the voltage measurements corresponds to a common level of precision among smart meters in the United States~\cite{EEI-AEIC-UTC:11}, whereas the case of 1\% noise represents the biggest error allowed in power systems. It can be observed that even under 1\% noise, all controllers under the proposed framework exhibit robustness, with higher level of communication leading to better capability for retaining noise-free performance under measurement noise.

\begin{table}[ht]
\centering
\caption{Accumulated cost under different levels of voltage magnitude measurement noise}
\begin{tabular}{c c c c c}
 \toprule
Noise & \texttt{Ctrl-NC} & \texttt{Ctrl-DC-1} & \texttt{Ctrl-DC-2} & \texttt{Ctrl-FC} \\
\midrule
0\%  & 0.7719  & 0.7575 & 0.7211 & 0.6987 \\
0.5\% & 0.7728 & 0.7580 & 0.7215 & 0.6987 \\
1.0\% & 0.7751 & 0.7598 & 0.7233 & 0.6994 \\
    \bottomrule
\end{tabular}
\label{tab:cost_noise}
\end{table}

\subsection{Discussion}
The simulation results above illustrate the effectiveness of the proposed framework, and indicate that stronger communication capability helps enhance the performance of the designed voltage controllers in the DG. Particularly, we have the following interpretations and remarks:
\begin{itemize}
    \item It can be observed that the \texttt{Ctrl-FC} still has a relatively large optimality gap compared to the \texttt{OPF} case. The reason is two-fold. First, even with full communication, the controllers still lack the information from the uncontrolled buses, while for the \texttt{OPF} case, the global information across the network is utilized. Second, the monotonicity requirement can introduce certain conservativeness when fitting to the OPF solutions. We note that there is a trade-off between theoretical guarantees (closed-loop stability) and the optimality of the controllers, as monotonicity is only a sufficient condition but not necessary for ensuring stability;  
    \item In the simulations above, we do not explicitly highlight the closed-loop stability guarantees of the proposed framework. We note that this has been showcased in our prior works~\cite{ZY-GC-MKS-JC:24-tps,JF-YS-GQ-SHL-AA-AW:24} that without monotonicity constraints on the controllers, the learned controller may easily lead to instability issues for the closed-loop system, especially when the voltage magnitudes are significantly higher or lower than the nominal value; 
    \item The test results above show that control performance improves as the level of communication increases. 
    However, precisely characterizing this dependency is theoretically challenging, not only in the present context, but in network control problems in general, and we leave it for future work. For the problem considered here, we note that a good starting point is considering the physical nature of distribution networks. It is known that buses that have further electrical distance to the substation may suffer from more serious voltage drops. These buses might not have enough control capability to regulate the voltage. Instead, buses that have a closer electrical distance to the substation do not contribute much to the voltage regulation because their local voltage might be satisfactory, thereby wasting their control capability. In such cases, it would be particularly helpful to deploy communication resources to the buses that have further electrical distance to the substation, especially adding communication links between such buses and buses with relatively strong control capability.
\end{itemize}

\section{Conclusions and Future Work}

We have presented a unified framework that allows us to design provably stable voltage controllers in distribution grids endowed with arbitrary communication infrastructures. The key enabler behind the proposed framework is that all local-acting controllers need to collectively satisfy a network-wide monotonicity condition, which ensures the closed-loop system is asymptotically stable. We have provided a design procedure that guarantees the derived local-acting controllers satisfy the network-wide monotonicity/stability condition under an arbitrary communication infrastructure and employed supervised learning to find the optimal ones. Simulation results with real-world data from the UCSD microgrid validate the effectiveness of the framework and reveal the role of communication in improving the control performance. Future work will include the theoretical characterization of the relationship between the connectivity of communication networks and control performance, which, from the experimental perspective, is that stronger connectivity implies better control performance. Other extensions include expanding the proposed framework to the case where active power is also controllable, adding safety limits of the voltage magnitudes during transient, considering inverter dynamics in the stability analysis, and characterizing the region of attraction when taking into account the nonlinear power flow model.

\bibliographystyle{ieeetr}

\appendix

\begin{table}[htb]
\centering
\caption{UCSD microgrid line parameters.}
 \begin{tabular}{c c c c c c c c}
 \toprule 
 From & To & $r_{mn} (\Omega)$ & $x_{mn} (\Omega)$ & From & To & $r_{mn} (\Omega)$ & $x_{mn} (\Omega)$ \\
 \midrule
 0 & 1 & $0.0174$ & $0.0002$ & 15 & 25 & $0.7802$ & $1.1703$ \\
 1 & 2 & $0.0232$ & $0.4855$ & 11 & 26 & $0.2722$ & $0.3906$ \\
 1 & 3 & $0.0238$ & $0.4894$ & 11 & 27 & $0.4659$ & $0.6989$ \\
 1 & 4 & $0.0232$ & $0.4778$ & 11 & 28 & $0.4659$ & $0.0002$ \\
 2 & 5 & $0.0185$ & $0.0278$ & 5 & 29 & $0.1126$ & $0.2166$ \\
 3 & 6 & $0.0255$ & $0.0382$ & 27 & 30 & $0.3496$ & $0.5244$ \\
 4 & 7 & $0.0324$ & $0.0486$ & 6 & 31 & $0.5803$ & $0.8704$ \\
 5 & 8 & $0.5452$ & $0.8178$ & 9 & 32 & $0.2547$ & $0.3820$ \\
 6 & 9 & $0.4376$ & $0.6563$ & 9 & 33 & $0.5093$ & $0.7640$ \\
 7 & 10 & $0.4352$ & $0.6529$ & 9 & 34 & $0.3478$ & $0.5218$ \\
 5 & 11 & $0.9417$ & $1.4125$ & 9 & 35 & $0.2722$ & $0.3906$ \\
 6 & 12 & $0.9000$ & $1.3500$ & 9 & 36 & $0.2722$ & $0.0002$ \\
 8 & 13 & $0.9000$ & $0.0005$ & 12 & 37 & $0.0324$ & $0.0486$ \\
 8 & 14 & $0.5961$ & $0.8942$ & 37 & 38 & $0.7802$ & $1.1703$ \\
 8 & 15 & $0.6338$ & $0.9506$ & 37 & 39 & $0.3700$ & $0.3717$ \\
 8 & 16 & $0.1922$ & $0.2882$ & 12 & 40 & $0.1105$ & $0.1658$ \\
 14 & 17 & $0.2275$ & $0.3412$ & 40 & 41 & $0.1109$ & $0.2132$ \\
 5 & 18 & $0.1160$ & $0.2230$ & 12 & 42 & $0.1922$ & $0.2882$ \\
 13 & 19 & $0.1268$ & $0.1901$ & 7 & 43 & $0.1175$ & $0.1762$ \\
 13 & 20 & $0.3478$ & $0.5218$ & 43 & 44 & $0.0191$ & $0.0286$ \\
 14 & 21 & $0.0174$ & $0.0260$ & 10 & 45 & $0.4457$ & $0.6685$ \\
 17 & 22 & $0.1256$ & $0.1884$ & 45 & 46 & $0.0868$ & $0.1302$ \\
 17 & 23 & $0.0174$ & $0.0260$ & 45 & 47 & $0.1105$ & $0.1658$ \\
 15 & 24 & $0.3531$ & $0.5296$ & 7 & 48 & $0.0301$ & $0.0451$ \\
 \bottomrule
 \end{tabular}
 \label{tab:microgrid_resistence_reactance}
 \end{table}

\end{document}